\documentclass{amsart}
\pdfoutput=1
\usepackage{amsmath}
  \usepackage{paralist}
  \usepackage{graphics}   \usepackage{epsfig} \usepackage{graphicx}  \usepackage{epstopdf} \usepackage[colorlinks=true]{hyperref}
         \hypersetup{urlcolor=blue, citecolor=red}

\DeclareGraphicsExtensions{.pdf,.png,.tikz}
\setkeys{Gin}{height=28mm,keepaspectratio}

\usepackage[rgb]{xcolor}
\usepackage{tikz}

\usepackage[T1]{fontenc}
\usepackage[utf8]{inputenc}
\usepackage{pgfplots}
\usepackage{grffile}
\pgfplotsset{compat=newest}
\usetikzlibrary{plotmarks}
\usepgfplotslibrary{patchplots}

\usepgfplotslibrary{external}

\usepackage{subcaption}
\usepackage{tikz}
\usepackage{paralist}
\usepackage{cite}
\usepackage{bm} \usepackage[mathscr]{eucal}
\usepackage{enumerate}
\usepackage{array}
\renewcommand{\arraystretch}{1.5} \usepackage{amsthm}

\usepackage{tikzscale}

\usepackage{mathtools}
\mathtoolsset{showonlyrefs}

\newcommand{\C}{\mathbb{C}}

\newcommand{\R}{\mathbb{R}}

\newcommand{\rT}{\top}

\newcommand{\id}{\mathrm{Id}}
\renewcommand{\phi}{\varphi}

\newcommand{\T}{\mathbb{T}}

\newcommand{\abs}[1]{ \ensuremath{\left\lvert{#1}\right\rvert}}

\DeclareMathOperator{\tr}{tr}
\DeclareMathOperator{\cof}{Cof}
\DeclareMathOperator{\sgn}{sgn}

\newcommand{\jac}{ \ensuremath{\nabla} }

\newcommand{\abc}{Arnold--Beltrami--Childress}
\newcommand{\owc}{Okubo--Weiss--Chong}
\newcommand{\ow}{Okubo--Weiss}
\newcommand{\kam}{Kolmogorov--Arnold--Moser}
\newcommand{\pf}{Perron--Frobenius}
\newcommand{\threed}{3D}
\newcommand{\twod}{2D}

\newcommand{\ftle}{Finite-Time Lyapunov Exponents}

\newcommand{\colrule}{\hline}
\usepackage{hyphenat}
\def\currentvolume{X}
 \def\currentissue{X}
  \def\currentyear{200X}
   \def\currentmonth{XX}
    \def\ppages{X--XX}
     \def\DOI{10.3934/xx.xx.xx.xx}

\newtheorem{theorem}{Theorem}[section]

\newtheorem{lemma}[theorem]{Lemma}

\theoremstyle{definition}
\newtheorem{definition}[theorem]{Definition}
\newtheorem{remark}{Remark}
\newtheorem*{remarkstep*}{Remark to Step}

 \title[Mesochronic classification in 3D vector fields]{Mesochronic classification of trajectories in incompressible 3D vector fields over finite times}

\author[
Budi\v{s}i\'{c},
Siegmund,
Son, and
Mezi\'{c}
]{}


 \email{marko@math.wisc.edu}
 \email{stefan.siegmund@tu-dresden.de}
 \email{dtson@math.ac.vn}
 \email{mezic@engr.ucsb.edu}

 \thanks{MB was funded under US Office of Naval Research MURI grant
   N00014-11-1-0087 and US National Science Foundation grant CMMI-1233935. SS
   was partly supported by the German Research Foundation (DFG) through the
   Cluster of Excellence (EXC 1056), Center for Advancing Electronics Dresden
   (cfaed). IM was funded under US Office of Naval Research MURI grant
   N00014-11-1-0087 and Army Research Office MURI grant W911NF-14-1-0359.}

\thanks{$^*$ Corresponding author: marko@math.wisc.edu}

\begin{document}
\maketitle

\centerline{\scshape Marko Budi\v{s}i\'{c}$^*$}
\medskip
{\footnotesize
 \centerline{Department of Mathematics}
   \centerline{University of Wisconsin, Madison}
   \centerline{ Madison, WI, USA}
}
\medskip

\centerline{\scshape Stefan Siegmund}
\medskip
{\footnotesize
 \centerline{Center for Dynamics \& Institute of Analysis}
   \centerline{Department of Mathematics, TU Dresden}
   \centerline{ Dresden, Germany}
}
\medskip

\centerline{\scshape Doan Thai Son}
\medskip
{\footnotesize
 \centerline{Department of Probability and Statistics}
   \centerline{Institute of Mathematics, Vietnam Academy of Science and Technology}
   \centerline{ Hanoi, Vietnam}
}
\medskip

\centerline{\scshape Igor Mezi\'{c}}
\medskip
{\footnotesize
 \centerline{Department of Mechanical Engineering}
   \centerline{University of California, Santa Barbara}
   \centerline{ Santa Barbara, CA, USA}
}
\bigskip

\begin{abstract}
The mesochronic velocity is the average of the velocity field along trajectories generated by the same velocity field over a time interval of finite duration. In this paper we classify initial conditions of trajectories evolving in incompressible vector fields according to the character of motion of material around the trajectory. In particular, we provide calculations that can be used to determine the number of expanding directions and the presence of rotation from the characteristic polynomial of the Jacobian matrix of mesochronic velocity. In doing so, we show that (a) the mesochronic velocity can be used to characterize dynamical deformation of three-dimensional volumes, (b) the resulting mesochronic analysis is a finite-time extension of the Okubo--Weiss--Chong analysis of incompressible velocity fields, (c) the two-dimensional mesochronic analysis from Mezi\'{c} et al. ``A New Mixing Diagnostic and Gulf Oil Spill Movement'', Science \textbf{330}, (2010), 486-–489, extends to three-dimensional state spaces. Theoretical considerations are further supported by numerical computations performed for a dynamical system arising in fluid mechanics, the unsteady Arnold--Beltrami--Childress (ABC) flow.
\end{abstract}

\bibliographystyle{AIMS}
\section{Introduction}
\label{sec:introduction}

Chaotic advection is the theory of material transport in fluids based in
dynamical systems theory.~\cite{Aref1984a, Ottino1989, Wiggins1992} It is
largely rooted in analysis of geometric structures in flows with a simple
time-dependence, time-autonomous or periodic. Since the realization that even
flows with a complicated time dependence, e.g., turbulent flows, possess organized Lagrangian structures, it has become increasingly important to detect
geometric structures analogous to invariant manifolds in steady flows. In
particular, detection of structures using only finite-time information about
the flow has been seen as the most practically-useful
direction.~\cite{Haller1998, Poje1999, Haller2001, Shadden2005, Froyland2010,
  Allshouse2012,Samelson2013}

The need for analysis of geometric structures that organize advection is not
purely academic. Transport of material by fluid flows played a crucial role in
the fallout from several recent catastrophic events, namely, the volcanic
eruption of Eyjafjallaj{\"o}kull (2010), the Deepwater Horizon oil spill
(2010), and the nuclear disaster in Fukushima (2011). These events highlighted
how important it is to detect organizing geometric structures in (near) real
time from data, either measured or generated by detailed simulation models;
consequently, such problems have become a very active intersection of
dynamical systems and fluid dynamics.

The problem of identifying geometric structures in flows has resulted in
several approaches, each focusing on a somewhat-different structure as the
objective of its analysis.

The theory of \emph{Lagrangian Coherent Structures}~\cite{Haller2000} (LCS)
identifies barriers that organize the transport in flows with complex
time-dependence. Initially, LCS were closely associated with computation of
Finite-Time Lyapunov Exponent fields~\cite{Shadden2005}; more recently, they have
been re-formulated using a variational principle~\cite{Haller2011a, Haller2012,
  Blazevski2014}, which defines them as certain geodesic lines of the local
deformation field induced by the fluid flow. This new definition allows a
finer classification of LCS, both in two and three dimensions, based on the
type of deformation, e.g., hyperbolic, elliptic, corresponding to different
behaviors of fluid parcels in the flow. The recent review by
Haller~\cite{Haller2015} gives a detailed coverage of the current
state-of-the-art of techniques centered around LCS.

LCS and associated theories mostly focused on magnitude of non-rotational
deformation in the flow. The rotation deformation has been classically studied
by Poincar\'e in topological dynamics and Ruelle~\cite{Ruelle1985} in ergodic
theory; recently, a Finite-Time Rotation Number~\cite{SzezechJr2013} has been proposed as a useful
quantity in analysis of flows. At the closing of this manuscript we were also
notified of the recent work by Farazmand and Haller~\cite{Farazmand2016},
working along the same lines.

In the effort to study the ``stretch-and-fold'' mechanism for chaos in finite
time, most studies focus on the first-order ``stretch''. Folding in finite
time has not received direct attention; an exception is the study of the
Finite-Time Curvature Fields~\cite{Ma2014, Ma2015, Ma2016}. At this point, structures
observed through all these methods have been connected mostly on
phenomenological basis, through comparison of visualizations, showing
considerable overlap between observed structures but, also, non-negligible
differences.

Magnitudes of the local material deformation are typically estimated by
processing velocity gradients; they cannot be precisely computed in the
absence of the detailed data about the velocity field, e.g., when the system
is sampled by sparse trajectories only. In sparsely-sampled planar systems,
trajectories can be represented by space-time braids --- extremely-reduced,
symbolic representations of trajectories. The resulting approach, known as
\emph{braid dynamics}~\cite{Boyland2000, Allshouse2012, Thiffeault2010} has
been successful in providing lower bounds on the amount of topological
deformation present in the flow, in limited-data settings. The obtained bounds
have been used both in design and analysis of the material advection;
unfortunately, there are currently no extensions of braid dynamics to
three-dimensional flows.

Instead of looking for barriers to transport, as is the case with the LCS
theory, the theory of \emph{almost-invariant sets} identifies a collection of
sets, fixed in space, such that the material placed in them does not leak
out. These sets act as routes for the material transport. The approach is
based on the \pf{} transfer operator, which models how the flow moves
distributions of points, instead of individual trajectories. The \pf{}
operator is always infinite-dimensional and linear; the identification of
almost-invariant sets is then intimately connected with approximating its
eigenfunctions. While the \pf{} operator has been a staple of the ergodic and
probability theory since the early 20th century, it was introduced to the
applied, non-probabilistic context by Dellnitz and Junge~\cite{Dellnitz1999,
  Dellnitz2002} as the basis for identification of invariant sets in
time-invariant dynamical systems. Since then, the theory has been expanded to
include detection of almost-invariant sets of autonomous
systems~\cite{Froyland2003, Froyland2009}, and flows with more general time
dependencies~\cite{Froyland2010, Froyland2010a}.

Spatial invariants of dynamical systems relate to infinite-time averages of
functions along Lagrangian trajectories. This connection between the ergodic
theory and applied mathematics was initially developed in
Ref.~\cite{Mezic1994}. Based on these ideas, Ref.~\cite{Poje1999}
proposed that even finite-time averages of functions can enable detection of
geometric structures important for fluid transport, which broadly constitutes
the \emph{mesochronic}, i.e., time-averaged, theory of transport in
fluids. The utility of time-averages has been corroborated on numerical and
experimentally-realizable flows with simple time dependence~\cite{Mezic1999,
  Malhotra1998, Mezic2002,Mancho2013,Madrid2009}.

The mesochronic techniques have developed in two directions. One focuses on
computations of ergodic invariant sets using long-time averages of a large set
of averaged functions.~\cite{Mezic1994, Mezic1999, Levnajic2010, Budisic2012b}
The other, which we follow here, does not aim to compute all invariant sets;
rather it uses much shorter averages of the velocity field itself to identify
the \emph{character} of the deformation, i.e., presence or absence of
rotation.~\cite{Mezic2010} This is in contrast with mentioned LCS and
rotational theories, which describe deformation starting from analysis of the
\emph{magnitude} of the deformation.

Before we dive into calculations, we give a short explanation of the
approach. Introductory courses in dynamical systems often discuss the
stability of a stagnation point \(p\) in a planar system \(\dot x = f(x)\) by
looking at its linearization \(\dot \xi = \jac f\vert_{p} \cdot \xi\) around a
fixed point \(p\) whose stability depends on positions of two eigenvalues of
the Jacobian \(\jac f\vert_{p}\). Instead of computing the eigenvalues, their
locations can be inferred from the trace and the determinant of \(\jac
f\vert_{p}\). If the flow is incompressible, the trace is zero, so the
determinant alone is needed for the full stability analysis. For unsteady
systems this analysis may not always hold; however, Ref.~\cite{Mezic2010}
showed that even then similar results can be obtained by looking at the
Jacobian matrix of the velocity \emph{averaged} over a Lagrangian trajectory,
termed \emph{the mesochronic Jacobian matrix}. Away from fixed points, the
calculation does not compute stability, but rather the spectral class of the
Jacobian: hyperbolic (strain) or elliptic (rotation), termed \emph{mesochronic
  classes}. Applied to prediction of the oil slick transport in the aftermath
of the Deep Water Horizon spill~\cite{Mezic2010}, mesochronic analysis showed
that regions of mesohyperbolicity correspond to jets which dispersed the
slick, while mesoelliptic zones correspond to centers of eddies in which the
slick accumulated.  The main contribution of this paper is the extension of
the mesochronic analysis to three-dimensional flows.

There are several connections of the mesochronic analysis with other
approaches.
\begin{compactitem}
\item On a fundamental level, averages of functions along trajectories are
  intimately related to spectral properties of the Koopman
  operator~\cite{Koopman1931, Mezic2004, Mezic2005}, which is adjoint to the
  mentioned \pf{} operator.
\item Greene~\cite{Greene1979,Greene1968} defined the \emph{residue criterion}
  in order to predict the order of destruction of \kam{} (KAM) tori in
  perturbed Hamiltonian maps. The computation of the residue is almost
  identical to that of the mesohyperbolicity indicator for \twod{} flows. The
  three-dimensional version of the residue criterion~\cite{Fox2013} also
  resembles mesochronic indicators introduced here.
\item An analysis of oceanic flows based on the Jacobian of the instantaneous,
  i.e., non-averaged, velocity is well-known as \owc{}
  partition~\cite{Okubo1970, Weiss1991, Chong1990}; this is the limit of the
  mesochronic theory as the averaging time \(T \to 0\).
\item In the other extreme, as \(T \to \infty\), real parts of eigenvalues of
  the mesochronic Jacobian relate to Lyapunov exponents and rotation numbers).
  It can be shown under generic conditions, using the polar
    decomposition of $\jac\psi_T$ that the limit of the eigenvalues of the
    gradient of the flow map are the Lyapunov exponents (see the heuristic
    discussion in~\cite{Goldhirsch1987}, that can be turned into rigorous
    proofs using ergodic-theoretic techniques in~\cite{Ruelle1985}, see
    also~\cite[Chapt.\ 3, \S 9]{Adrianova1995})
\item Finally, a recent inquiry~\cite{Farazmand2015a} into connections with
  Lagrangian Coherent Structures, in a form related to an earlier work in
  Ref.~\cite{Poje1999}, showed that the LCS techniques are capable of
  uncovering some of the boundaries between mesochronic classes.
\end{compactitem}

The paper is organized as follows. In Section~\ref{sec:preliminaries} we
introduce the precise definitions for dynamical systems we are considering,
review the basics of differential geometry needed, and re-state the \owc{}
analysis in these terms. Section~\ref{sec:mesochronic} contains the main
result, the \threed{} mesochronic classification, while
Section~\ref{sec:limits-boundary-cases} makes connections to Lyapunov, \owc{},
and \twod{} mesochronic analyses. In Section~\ref{sec:computation} we
illustrate the technique by a set of analytic and numerical examples; in
particular the steady and unsteady \abc{} flow~\cite{Dombre1986}. Numerical
details are given in the Appendices. The paper closes with the discussion in
Section~\ref{sec:discussion}.

\section{Preliminaries}
\label{sec:preliminaries}

Consider a time-varying differential equation
\begin{equation}\label{eq:dynamical-system}
  \dot x = f(t,x)
\end{equation}
with a $C^3$ velocity field $f : D \subset \R \times \R^3 \rightarrow
\R^3$. For an initial condition $(t_0,x_0) \in D$ let $t \mapsto
\phi(t,t_0,x_0)$ denote the solution (or trajectory) of the initial value
problem~\eqref{eq:dynamical-system}, with $x(t_0)=x_0$. Throughout the paper
we assume for an arbitrary but fixed \emph{initial time} $t_0 \in \R$,
\emph{duration} $T>0$ and open \emph{set of initial values} $X(t_0) \subset
\R^3$ that \(\forall x_0 \in X(t_0)\), \(t \mapsto \phi(t,t_0,x_0)\)
exists on the whole interval \(I := [t_0,t_0 + T]\).

For $t \in I$ define $X(t) := \{\phi(t,t_0,x_0) \in \R^3 : x_0 \in X(t_0)\}$
and $X := \{(t,x) \in \R \times \R^3 : t \in I, x \in X(t)\}$. Then $X \subset
D$ by assumption und for $t_1 \in I$ the map $\phi(t_0+T,t_0, \cdot) : X(t_0)
\rightarrow X(t_1)$ is well-defined. In particular, we define the
\emph{time-$T$ map}
\begin{equation}\label{eq:time-t}
  \begin{aligned}
    \psi_{T} &: X(t_0) \rightarrow X(t_0 + T), \\
    \psi_{T}(x) &:= \phi(t_{0}+T,t_{0},x),
  \end{aligned}
\end{equation}
usually called \emph{Poincar\'{e} map} if the equation is periodic, i.e., if
for the chosen \(T\), $f(t,x) \equiv f(t+T,x)$.

We are mainly interested in \emph{finite-time} dynamics for a fixed duration
$T>0$ but will also investigate the \emph{instantaneous} (in the zero-time
limit \(T \to 0^{+}\)) and \emph{asymptotic} (in the infinite-time limit \(T
\to +\infty\)) dynamics, assuming the solution \(\phi(\cdot,t_0,x_0)\) exists on
\([t_0,\infty)\).

\emph{Observables} of~\eqref{eq:dynamical-system} are continuous functions \(F
: X \to \R^{n}\), which are evaluated along arbitrary solutions $t \mapsto
x(t)$ on $I$. They are used to model physical measurements of a state of the
system, e.g., the time trace \(t \mapsto F\left(t, x(t) \right) \) might
represent the ocean temperature recorded by a sensor as it is passively
carried by ocean currents along the trajectory \(x(t)\). A \emph{time average}
or \emph{trajectory average} $\tilde{F}_{T} : X(t_0) \rightarrow \R$ of an
observable \(F\) on $I = [t_0,t_0+T]$, defined by
\[
\tilde{F}_{T}(x_{0}) := \frac{1}{T}\int_{t_{0}}^{t_{0} + T} F( \tau, x(\tau) ) d\tau,
\]
is a function that depends on the initial value $x(t_0) = x_0$ of the
trajectory \(x(t)\) at time $t_0$. Trajectory averages depend on the duration
$T>0$ and can be analyzed from the instantaneous, asymptotic, or finite-time
perspective. The instantaneous case is the most obvious, as \( \lim_{T \to 0}
\tilde{F}_{T}(x_{0}) = F(t_{0},x_{0})\), e.g., if \(F\) is a component of
the velocity field itself. Certain choices of \(F\) can still provide valuable
information, as we explain in the next paragraph.  Asymptotic analysis studies
\emph{ergodic averages}, i.e., limits $\tilde{F}_\infty(x_{0}) := \lim_{T \to
  \infty} \tilde{F}_{T}(x_{0})$, in case they exist. Ergodic theory analyzes
limits $\tilde{F}_\infty$ of observables \(F\) which are specified only in
general terms, e.g., only by the space of functions from which they are
drawn. Even in such general cases valuable information can be recovered, e.g.,
on time-invariant measures on the state space.~\cite{Mezic1999,Budisic2012b}

On the other hand, choosing a particular observable can provide us with more
detailed information. Since the components of the velocity field \(f(t,x) =
[f_{1}(t,x), \allowbreak f_{2}(t,x), \allowbreak f_{3}(t,x)]^{\rT}\) are themselves continuous
functions on the time-state space \(X \subset \R \times \R^{3}\), they are
observables.  One could argue that they are the most distinguished observables
for analysis of dynamical systems, as they directly provide dynamical
information about the behavior of the system. We adopt the
\emph{velocity-as-observable} viewpoint and analyze the time average of the
velocity field, which was termed \emph{mesochronic
  velocity}~\cite{Mezic2010}. We note that the values of the mesochronic
velocity appear as quantities of interest in~\cite{Poje1999}
while~\cite{Mezic2010} is the first to look into their gradients.

\begin{definition}[Mesochronic velocity and mesochronic Jacobian]
  The \emph{mesochronic velocity} \(\tilde{f} : X(t_0) \subset \R^3
  \rightarrow \R^3\) of~\eqref{eq:dynamical-system} on \( I =
  [t_{0},t_{0}+T]\) is given by
  \begin{equation}
    \tilde{f}_{T}(x) := \frac{1}{T}\int_{t_{0}}^{t_{0} + T} f\left(\tau,\phi(\tau, t_{0}, x)\right) d\tau.
  \end{equation}
  The Jacobian matrix \(\jac \tilde{f}_{T}\) containing partial spatial derivatives \(
    {[\jac \tilde{f}_{T}]}_{ij} \allowbreak:=
      \partial {[\tilde{f}_{T}]}_{i} / \partial x_{j}
\)
  is termed \emph{the mesochronic Jacobian}.
\end{definition}
Note that the spatial derivatives and the averaging over trajectories do not commute, i.e., \(\jac \tilde{f}_{T}\) is not equal to the average of the instantaneous Jacobian over trajectories.

The mesochronic velocity in the instantaneous limit \(\tilde{f}_{T}
\xrightarrow{T \to 0^{+}} f\) coincides with the velocity field $f$.  The
asymptotic limit for $T \to \infty$ exists in many cases, for example, if the
dynamical system is autonomous and volume-preserving on a compact domain.

We use the mesochronic velocity to determine the character of the evolution of
a material volume (see Section~\ref{sec:diffeomorphism}) by an incompressible
dynamical system, which satisfies the Liouville condition
\begin{equation}
  (\nabla \cdot f)(t,x) = \tr \jac f(t,x) \equiv 0.
  \label{eq:incompressible-velocity-field}
\end{equation}

Although the limits \(T\to0^{+}\) and \(T\to\infty\) have been studied
classically, neither theory is applicable to transient behavior. On the other
hand, the finite-time analysis of the mesochronic velocity recovers the
character of the time-\(T\) map \(\psi : X(t_0) \rightarrow X(t_0+T)\), which
captures transients at the time-scale \(T\)~\cite{Mezic2010}.

\subsection{Deformation of a volume cell under a diffeomorphism}
\label{sec:diffeomorphism}

We now briefly review basic differential geometry that characterizes
deformation under a volume-preserving diffeomorphism \(\psi\) using the
spectral class of its Jacobian \(\jac \psi\). This theory is later applied to
time-\(T\) maps \(\psi_{T}\) of dynamical systems over finite time intervals.

Let \(\psi:U \to V\) be a diffeomorphism between two open subsets \(U \subset
\R^{3}\), \(V \subset \R^{3}\), with the usual volume measure on
\(\R^{3}\). We are interested in deformation of an infinitesimal volume cell
surrounding \(x \in U\) as \(x \mapsto \psi(x)\). The central object of our
interest is the Jacobian matrix \(\jac \psi: U \to \R^{3 \times 3}\).
It is a basic result in differential geometry that volumes of a set \(S
\subset U\) and its image \(\psi(S) \subset V\) are equal if and only if
\(\abs{\det \jac \psi(x)} \equiv 1\). We now restrict our attention to
\emph{orientation- and volume-preserving} maps \(\psi\), i.e., maps for which
\(\det \jac \psi \equiv 1\).

At the coarsest level, we distinguish between a \emph{hyperbolic} deformation,
when the volume cell is \emph{deformed along all three spatial dimensions},
and the opposite, non-hyperbolic character. Let \(\mu_{1}, \mu_{2}, \mu_{3}
\in \C\) denote the eigenvalues of \(\jac \psi\), assuming
\(\abs{\mu_{1}}\leq\abs{\mu_{2}}\leq\abs{\mu_{3}}\). Different fields of
mathematics may interpret presence or absence of hyperbolicity differently,
e.g., as the material deformation in continuum mechanics, or stability of the
map in dynamical systems and control. Since our analysis could include both
domains, we refer to presence/absence of hyperbolicity, and their
sub-classification (see below) as the \emph{spectral character} of the
diffeomorphism.

\begin{definition}[Hyperbolicity]\label{def:hyperbolicity}
  The map \(\psi\) is \emph{hyperbolic at \(x\)} if no eigenvalues \(\mu_{1},
  \mu_{2}, \mu_{3}\) of the Jacobian \(\jac \psi(x)\) lie on the unit circle
  in the complex plane, i.e., \(\forall\ i = 1,2,3\), \(\abs{\mu_{i}} \not =
  1\). Otherwise, it is \emph{non-hyperbolic at \(x\)}. In particular, if
  \emph{all} eigenvalues lie on the unit circle, i.e., \(\forall\ i = 1,2,3\),
  \(\abs{\mu_{i}} = 1\), the non-hyperbolic map is \emph{elliptic}.
\end{definition}

Depending on the complexity of eigenvalues, we further distinguish
\emph{sellar} (saddle-like) and \emph{helical} (spiral/helix-like) character
of the deformation under \(\psi\).\footnote{From Latin, \emph{sella}, saddle,
and \emph{helix}, spiral.}

\begin{definition}[Sellar and helical deformation]\label{def:sellicity-helicity}
  The map \(\psi\) is \emph{sellar at \(x\)} if all eigenvalues \(\mu_{1},
  \mu_{2}, \mu_{3}\) of the Jacobian \(\jac \psi(x)\) are real and
  non-defective (their algebraic and geometric multiplicities match). If,
  instead, there is a pair of complex-conjugate eigenvalues, the map is
  \emph{helical} (spiral-like) at \(x\).
\end{definition}

Deformation at a sellar point exhibits three distinct spatial axes, meeting at \(x\),
whose directions are preserved under \(\psi\). The directions of preserved
spatial axes correspond to real-valued eigenvectors of \(\jac \psi\), while
the associated (real) eigenvalues \(\mu\) determine whether the points along
the axes are moving away from \(x\), \((\abs{\mu} > 1)\), moving towards
\(x\), \((\abs{\mu} < 1)\), or remain neutral \((\abs{\mu} = 1)\).

Deformation at a helical point exhibits only a single preserved spatial axis,
around which a volume cell is rotated, resulting in a single real-valued
eigenvalue. As \(\jac \psi(x)\) is a real matrix, any complex eigenvalues must
arise in conjugate pairs, \(\mu,\,\bar{\mu}\). The modulus of the complex pair
again determines expansion or contraction of the material in the rotation
plane, while the real and imaginary components of the associated eigenvector
pair span the rotation plane.

Since we restrict our analysis to volume-preserving maps \(\psi\), existence
of an eigenvalue inside the unit circle (contraction) necessarily means that
at least one other eigenvalue lies outside of the unit circle (expansion), as
\(\abs{\det \jac \psi} = \abs{\mu_{1} \mu_{2}\mu_{3}}\equiv
1\). All possible combinations are enumerated in
Definition~\ref{def:spectral-classes} in which we use the symbols \(+\) and
\(-\) to denote expansion and contraction directions, respectively.
\begin{definition}[Spectral classes]\label{def:spectral-classes}
  Let \(\psi:U \to V\) be a volume- and orientation-preserving diffeomorphism,
  \(x \in U\) a point, and \(\jac \psi(x)\) the Jacobian of \(\psi\) at \(x\)
  with eigenvalues \(\mu_{i}\), ordered as \(\abs{\mu_{1}} \leq \abs{\mu_{2}}
  \leq \abs{\mu_{3}}\). The class of the point \(x\) is determined according to Table~\ref{tab:spectral-classes} by the number of eigenvalues of \(\jac\psi(x)\) that are inside and on the unit circle.
\end{definition}

\begin{table}
  \begin{tabular}{ m{2.25cm} m{3.5cm} }
    \parbox[t]{2cm}{\textbf{Class}} & \parbox[t]{4cm}{\textbf{Condition}} \\ \colrule
    \(\bm{[- + +]}\) sellar (hyp.~saddle) & \parbox[t]{4cm}{\(\mu_{1,2,3} \in \R\) and \\ \(\abs{\mu_{1}} < 1 < \abs{\mu_{2}} < \abs{\mu_{3}}\),} \\ \colrule
    \(\bm{[- - +]}\) sellar (hyp.~saddle) & \parbox[t]{4cm}{\(\mu_{1,2,3} \in \R\) and \\ \(\abs{\mu_{1}} < \abs{\mu_{2}} < 1 < \abs{\mu_{3}}\),} \\ \colrule
    \(\bm{[- + +]}\) helical (hyp.~spiral) & \parbox[t]{4cm}{\(\mu_{1} \in \R,\ \mu_{2} = \bar \mu_{3} \not \in \R\) and \\  \(\abs{\mu_{1}} < 1 < \abs{\mu_{2}} = \abs{\mu_{3}}\),} \\ \colrule
    \(\bm{[- - +]}\) helical (hyp.~spiral) & \parbox[t]{4cm}{\(\mu_{1} = \bar \mu_{2} \not \in \R,\ \mu_{3} \in \R\) and \\  \(\abs{\mu_{1}} = \abs{\mu_{2}} < 1 < \abs{\mu_{3}}\),}\\ \colrule
    Neutral saddle & \parbox[t]{4cm}{\(\mu_{1,2,3} \in \R\) and \\ \(\abs{\mu_{1}} < 1 = \abs{\mu_{2}} < \abs{\mu_{3}}\), }\\ \colrule
    Neutral helix & \parbox[t]{4cm}{\(\mu_{1} = \bar \mu_{2} \not \in \R,\ \mu_{3} \in \R\) and \\ \(\abs{\mu_{i}} \equiv 1\), }\\ \colrule
    Pure shear & \parbox[t]{4cm}{\(\mu_{1} = \mu_{2} = \pm 1\), \(\mu_{3} = 1\) and \\ \(\jac \psi\) is defective,} \\ \colrule
    Pure reflection & \parbox[t]{4cm}{\(\mu_{1,2,3} \in \{-1,1\}\) and \\ \(\jac \psi\) is not defective. }
  \end{tabular}
  \caption{Classification of a \threed{} volume-preserving diffeomorphism, depending on locations of eigenvalues  \(\mu_{1,2,3}\) of the Jacobian matrix (see Definition~\protect\ref{def:spectral-classes}). A matrix is defective, or non-diagonalizable, when it has less than 3 linearly independent eigenvectors.}\label{tab:spectral-classes}
\end{table}

The first four classes are hyperbolic, whereas the remaining cases are
non-hyperbolic. Informally, we will refer to signatures \(\bm{[- + +]}\) and
\(\bm{[- - +]}\) of hyperbolic points as, respectively, \emph{flattening} and
\emph{elongating}, due to the shape of a volume cell after application of
\(\psi\), as sketched in Figure~\ref{fig:classes-sketch}.
\begin{figure}[!ht]
  \begin{subfigure}[t]{0.3\linewidth}\centering
    \includegraphics{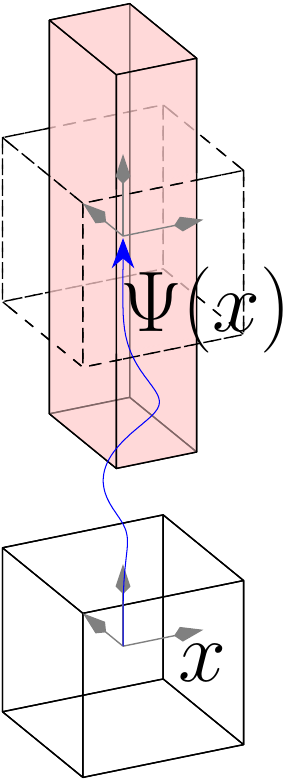}
    \caption{\(\bm{[- - +]}\) sellar deformation (elongating saddle)}
  \end{subfigure}\quad
  \begin{subfigure}[t]{0.3\linewidth}\centering
    \includegraphics{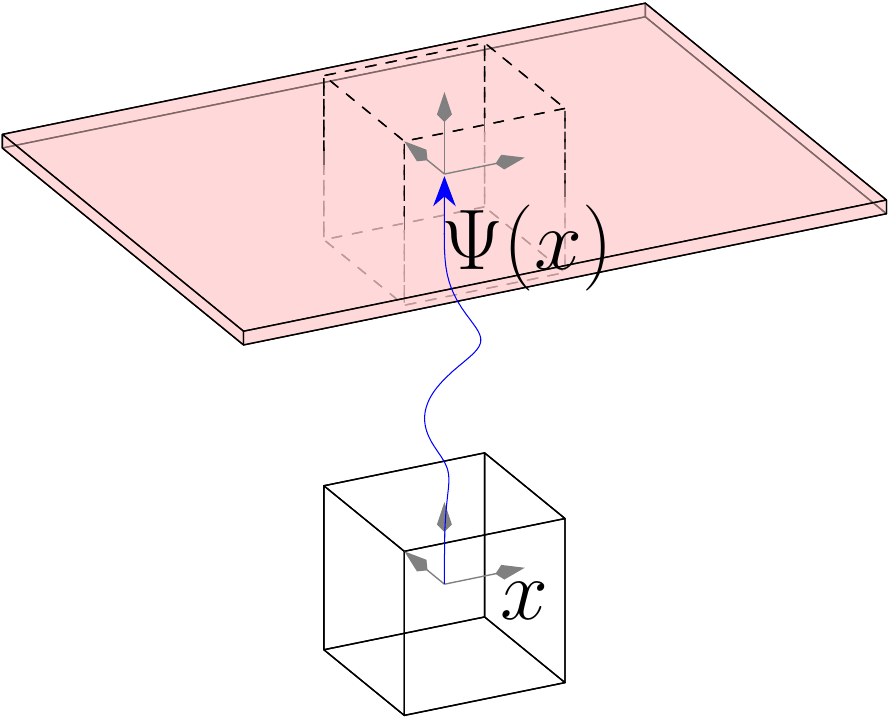}
    \caption{\(\bm{[- + +]}\) sellar deformation (flattening saddle)}
  \end{subfigure}\quad
  \begin{subfigure}[t]{0.3\linewidth}\centering
    \includegraphics{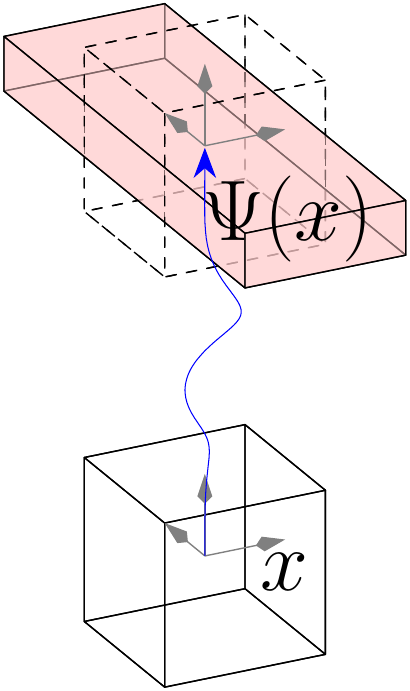}
    \caption{Neutral sellar (neutral saddle)}
  \end{subfigure}\\
  \begin{subfigure}[t]{0.3\linewidth}\centering
    \includegraphics{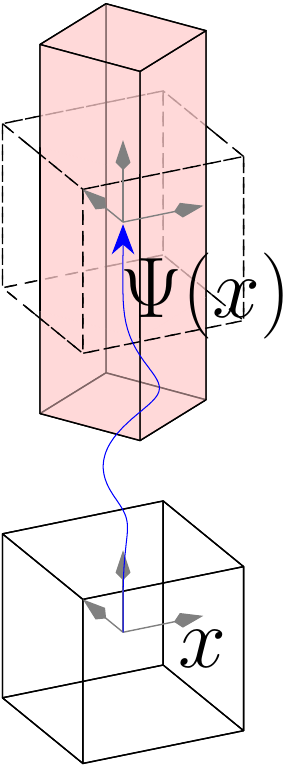}
    \caption{\(\bm{[- - +]}\) helical deformation (elongating spiral)}
  \end{subfigure}\quad
  \begin{subfigure}[t]{0.3\linewidth}\centering
    \includegraphics{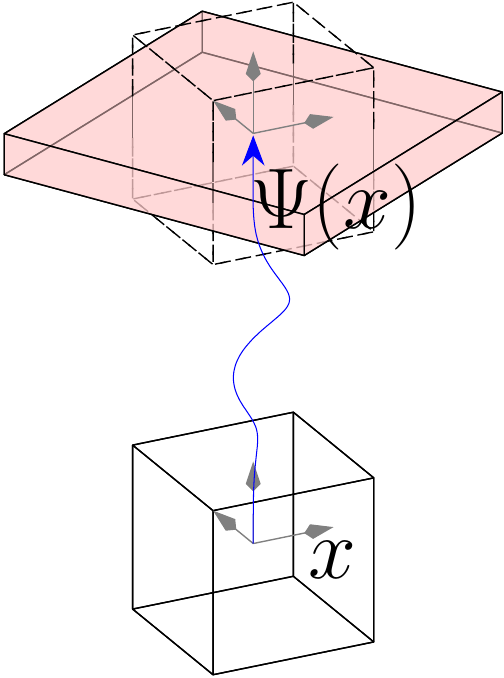}
    \caption{\(\bm{[- + +]}\) helical deformation (flattening spiral)}
  \end{subfigure}\quad
  \begin{subfigure}[t]{0.3\linewidth}\centering
    \includegraphics{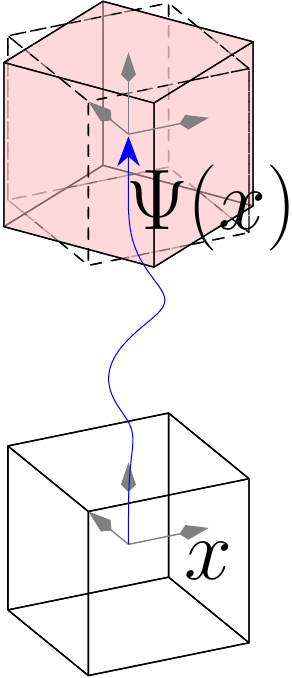}
    \caption{Neutral helical (neutral spiral)}
  \end{subfigure}
  \caption{\label{fig:classes-sketch}Sketches of deformation depending on the spectral class of a volume-preserving diffeomorphism \(\psi\) at the point  \(x \in U\)  (see
    Definition~\protect\ref{def:spectral-classes}). Note that the initial and final
    axes do not need to be parallel in general. Pure shear and reflection (including the identity) are not sketched.}
\end{figure}

\subsection{Instantaneous deformation by a dynamical system}
\label{sec:instantaneous}

We now apply the classification outlined in Section~\ref{sec:diffeomorphism}
to the time-\(T\) map \(\psi_T\) defined in~\eqref{eq:time-t} for a fixed time
interval \(T\) to recover the Okubo--Weiss--Chong~\cite{Okubo1970,Weiss1991,Chong1990} criterion for classification of velocity fields. Instead of forming the classification of \(\psi_{T}\) based on
properties of \(\jac \psi_{T}\), we chose instead to base our approach on properties of
the mesochronic Jacobian \(\jac \tilde{f}_{T}\), resulting in criteria whose expressions are consistent across all time scales \(T\).

\emph{Instantaneous classification} refers to the infinitesimally small time
interval length \(T\) for which the class of \(x \in \R^{3}\) under
\(\psi_{T}\) can be inferred from the spectrum of the Jacobian \(\jac
f(t_{0},x)\), through the connections with \(\jac \psi_{T}\). Both Jacobians are \(3\times 3\) matrices; for a general such matrix \(M\), the characteristic polynomial is given by
\begin{align}
  P_{M}(s) &= \det\left(s\id - M\right) = s^{3} - s^{2} \tr M  + s \tr\cof M + \det M \\\shortintertext{where}
  \tr M &= \sum_{i}s_{i},\qquad
  \det M = \prod_{i}s_{i},\\
  \tr\cof M &= \sum_{i} \frac{\det M}{s_{i}} = \frac{1}{2}( \tr M^{2} - (\tr M)^{2} ) \label{eq:cofactortrace}
\end{align}
The coefficient \(\tr\cof M\) is the cofactor trace, also called ``second trace'', in~\cite{Fox2013}; intuitively, it is the ``variance'' of the vector containing eigenvalues. Additionally, it follows from~\eqref{eq:cofactortrace} that when \(\det M = 1\), the second trace is equal to the trace of matrix inverse \(\tr \cof M = \sum s_{i}^{-1} = \tr M^{-1}\). Further spectral properties of such matrices are summarized in the Appendix~\ref{sec:cubic-poly-mat}.

As mentioned in Definition~\ref{def:spectral-classes}, spectral class of the time-\(T\) map is determined by locations of eigenvalues of \(\jac \psi_{T}\) which, in
turn, are roots of the characteristic polynomial
$P_{\psi} (\mu) = \det \left( \mu \id - \jac \psi_{T}(x) \right)$.

Expanding \(\jac \psi(x)\) into a Taylor series around \(T = 0\) yields
\begin{equation}
  \begin{aligned}
    \jac \psi_{T}(x) &= \id + T \jac f(t_{0},x) + \mathcal{O}(T^{2}) \\
    &\approx \id + T \jac f(t_{0},x),
  \end{aligned}\label{eq:timeT-velocity-jacobian-approx}
\end{equation}
for \(T \approx 0\). Consequently
\begin{align*}
  P_{\psi}(\mu) &= \det\left[(\mu-1)\id - T \jac f(t_{0},x)\right] \\
  &= T^{3}P_{f}\left( \frac{\mu - 1}{T} \right),
\end{align*}
for \(T \approx 0\), where \(P_{f} (\lambda) = \det \left( \lambda \id - \jac f(t_{0},x) \right) = \lambda^3-t_{f}\lambda^2+m_{f}\lambda-d_{f}\) is the characteristic polynomial of the Jacobian matrix \(\jac f\) of the instantaneous velocity field. In all cases, Jacobian matrices are evaluated at the time-space point \((t_{0},x)\) which is suppressed in notation.

Incompressibility of the flow implies
\begin{equation}
t_{f} \equiv 0 \text{\quad and \quad} d_{\Psi} \equiv 1, \label{eq:incompressibility-flow}
\end{equation}
and therefore
the spectral character depends only on the determinant \(d_{f}\) and the sum
of minors \(m_{f}\) of the velocity field Jacobian \(\jac f\). We now state
the \owc{} criteria, using terminology specified in
Definition~\ref{def:spectral-classes}.

\begin{theorem}[\owc{}~\cite{Chong1990}]\label{thm:3D-okubo-weiss}
  From the Jacobian of the velocity field (velocity gradient tensor) \(\jac
  f(t_{0},x)\), compute its determinant \(d_{f} = \det \jac f(t_{0},x)\), and
  cofactor trace \(m_{f} = \tr \cof \jac f(t_{0},x)\).  For \(T
  \to 0\), the point \((t_{0},x) \in I \times \R^{3}\) is hyperbolic if and
  only if
  \begin{equation*}
    d_{f} \not = 0.
  \end{equation*}
  Hyperbolic points can be further classified into four subclasses, depending
  on signs of the determinant \(d_{f}\) and the quantity
  \(27d_{f}^2+4m_{f}^3\), as listed in Table~\ref{tab:OWC}.
\end{theorem}
\begin{table}
  \begin{tabular}{ c  c  l }
    \(d_{f}\)   & \(27d_{f}^2+4m_{f}^3\)  & \textbf{OWC class} \\ \colrule
    \(-\)            & \(-\)            & \(\bm{[- + +]}\) saddle \\ \colrule
    \(-\)            & \(+\)            & \(\bm{[- + +]}\) helix \\ \colrule
    \(+\)            & \(-\)            & \(\bm{[- - +]}\) saddle \\ \colrule
    \(+\)            & \(+\)            & \(\bm{[- - +]}\) helix   \end{tabular}
\caption{\label{tab:OWC}\owc{} classification based on signs of quantities in the first two columns (see Theorem~\protect\ref{thm:3D-okubo-weiss}).}
\end{table}
The analogous result for planar dynamics is known as the \ow{}
criterion~\cite{Okubo1970,Weiss1991}.

\section{Mesochronic classification}
\label{sec:mesochronic}

The \owc{} criterion (Theorem~\ref{thm:3D-okubo-weiss}) classifies the points
in the domain of the time-\(T\) map \(\psi_{T}\) according to the spectral
character of the Jacobian \(\jac f\) of the velocity field. In general,
the spectral character of \(\jac \psi_{T}\) for \(T\) away from \(0\) will not be
approximated well by the spectral character of \(\jac f\), except in
extremely slowly varying flows. In order to capture both deformation at small and large \(T\), we replace \(\jac f\) by the Jacobian \(\jac \tilde f_{T}\) of the mesochronic velocity \(\tilde f_{T}\), defined as the Lagrangian average of the velocity \(f\) over the duration \(T\):
\begin{equation}
\label{eq:mesochronic-velocity}
\tilde{f}_{T}(x) := \frac{1}{T}\int_{t_{0}}^{t_{0} + T} f\big(\tau,\phi(\tau, t_{0}, x)\big) d\tau.
\end{equation}

To see why \(\jac \tilde{f}_{T}(x)\) plays an important role in our analysis,
we write the solution \(\phi(t,t_0,x)\) in its integral form
\begin{equation}
\phi(t,t_{0},x) = x + \int_{t_{0}}^t f\big(\tau,\phi(\tau,t_0,x)\big) d\tau,\label{eq:evolution}
\end{equation}
which is a consequence of the fundamental theorem of calculus. The same
integral appears both in~\eqref{eq:mesochronic-velocity}
and~\eqref{eq:evolution}, which we use to write the time-\(T\) map as
\begin{equation}
    \psi_{T}(x) = \phi(t_{0}+T,t_0,x) = x + T \tilde{f}_{T}(x),\label{eq:timeTmap}
\end{equation}
and hence
\begin{equation}
  \jac \psi_{T}(x) = \id + T \jac \tilde{f}_{T}(x). \label{eq:timeT-mh-j-link}
\end{equation}
Comparing with~\eqref{eq:timeT-velocity-jacobian-approx}, \(\jac
\psi_{T}(x) \approx \id + T \jac f(t_{0},x)\) which is valid for \(T \approx
0\), we see that for general time intervals $[t_0,t_0+T]$ the mesochronic
velocity Jacobian $\jac \tilde{f}_{T}(x)$ plays the same role as the velocity
field Jacobian $\jac f(t_{0},x)$ did for intervals of infinitesimal length.

In analogy to \owc{}, we classify \(x \in X(t_{0}) \subset
\mathbb{R}^{3}\) under the action \(\psi_{T}\) using only information obtained
from the mesochronic velocity and its Jacobian \(\jac \tilde{f}_{T}(x)\).
\begin{definition}[Mesochronic classes]\label{def:mesohyperbolicity}
  Given a fixed time interval \([t_{0},t_{0}+T]\), the point \(x \in X(t_{0})
  \subset \R^{3}\) is \emph{mesohyperbolic} if it is hyperbolic with respect
  to the diffeomorphism \(\psi_{T}\) (time-\(T\) map), i.e., no eigenvalues of
  \(\jac \psi_{T}(x)\) lie on the unit circle in the complex plane.
  Otherwise, \(x\) is \emph{non-mesohyperbolic}.

  Similarly, the \emph{mesochronic class} of \(x\) is the spectral class
   of \(\psi_{T}\) at \(x\), as specified by
  Definition~\ref{def:spectral-classes}.
\end{definition}

As explained in Section~\ref{sec:diffeomorphism}, the Liouville
incompressibility results in only two, instead of three, ``axes'' in which we understand the spectral class of \(\psi_{T}\):
\begin{compactenum}
\item number of contracting directions, indicated by labels \(\bm{[- + +]}\)
and \(\bm{[- - +]}\).
\item presence of a rigid rotation, indicated by helix vs.~saddle split.
\end{compactenum} The four mesohyperbolic classes are formed by choosing an
option along each of the axes above, with the non-hyperbolic classes acting as
boundary cases between them.

For conceptual and computational reasons we will formulate quantities
\(\Sigma\) and \(\Delta\), each corresponding to one of the ``axes''
above, such that the signs of their values sort the finite-time
dynamics around \(x\) into mesochronic classes. The number of contracting directions will be detected by \(\Sigma\), the presence of rotation by \(\Delta\).
\begin{theorem}[Mesochronic classification]\label{thm:mesochronic-classification}
  Let \(\jac \tilde{f}_{T}(x)\) be the mesochronic Jacobian matrix for the
  dynamics at the point \(x \in X(t_{0}) \subset \R^{3}\) and for the time
  interval \([t_{0},t_{0}+T]\) and
  \(
  P_{\tilde{f}}(\lambda) := \lambda^{3} - t_{\tilde{f}} \lambda^{2} + m_{\tilde{f}} \lambda - d_{\tilde{f}}
  \), the characteristic polynomial of the matrix \(\jac \tilde{f}_{T}(x)\).
  Define  \begin{equation}
    \begin{aligned} \Sigma :=&\
\frac{d_{\tilde{f}}T^{3}}{8-2m_{\tilde{f}}T^{2}-3 d_{\tilde{f}}T^{3}} \\
\Delta :=& -4 d_{\tilde{f}}^{4} T^{12} - 12 d_{\tilde{f}}^{3}m_{\tilde{f}}
T^{11} \\ &- 13 d_{\tilde{f}}^{2}m_{\tilde{f}}^{2}T^{10} -
6d_{\tilde{f}}m_{\tilde{f}}^{3}T^{9} \\ &+ (18 d_{\tilde{f}}^{2}m_{\tilde{f}}
- m_{\tilde{f}}^{4})T^{8} + 18 d_{\tilde{f}}m_{\tilde{f}}^{2} T^{7} \\ &+ (27
d_{\tilde{f}}^{2} + 4 m_{\tilde{f}}^{3})T^{6}.
    \end{aligned}
    \label{eq:sigmadelta}
  \end{equation}
  If \(8 - 2 m_{\tilde{f}}T^{2} - 3d_{\tilde{f}} T^{3} = 0\), then \(x\) is
  non-mesohyperbolic of the neutral-saddle type.

  If \(\Sigma\) is finite, the point \(x\) is classified into mesochronic
  classes according to Table~\ref{tab:mesochronic-classes}.  Mesohyperbolic
  classes are those for which \(\Sigma\) is finite and \(\Sigma \not =0\), i.e., for
  which both
  \begin{equation} d_{\tilde{f}} \not = 0 \text{ and } 8 - 2
m_{\tilde{f}}T^{2} - 3d_{\tilde{f}} T^{3} \not = 0.
    \label{eq:mesohyperbolicity}
  \end{equation}
\end{theorem}
The distinction between the shear and the boundary saddle class when \(\Delta = 0\) depends on eigenvectors of the mesochronic Jacobian and cannot be made purely based on the spectrum.

\begin{table}
  \centering
\begin{subfigure}[c]{0.3\linewidth}
\begin{tikzpicture}
\draw[->,gray] (-2,0) -- (2,0) node[anchor=north,color=black] {$\Delta$};
\draw[->,gray] (0,-2) -- (0,2) node[anchor=east,color=black] {$\Sigma$};
\node[align=center,font=\footnotesize] at (1,1) { $[- - +]$ \\ helix };
\node[align=center,font=\footnotesize] at (-1,1) { $[- - +]$ \\ saddle };
\node[align=center,font=\footnotesize] at (-1,-1) { $[- + +]$ \\ saddle };
\node[align=center,font=\footnotesize] at (1,-1) { $[- + +]$ \\ helix };
\node[align=center,font=\footnotesize] at (1,0) { neutral helix};
\node[align=center,font=\footnotesize] at (-1,0) { neutral saddle};
\node[align=center,font=\footnotesize,rotate=90] at (0,1) { \scriptsize shear or saddle };
\node[align=center,font=\footnotesize,rotate=90] at (0,-1) { \scriptsize shear or saddle };
\end{tikzpicture}
\end{subfigure}\qquad
\begin{subfigure}[c]{0.55\linewidth}
  \begin{tabular}[c]{ c c c }
    \(\Sigma\)  & \(\Delta\)  & \textbf{Mesochronic class} \\ \colrule
    \(-\)            & \(-\)            & \(\bm{[- + +]}\) mesosellar \\ \colrule
    \(-\)            & \(0\)            & \(\bm{[- + +]}\) mesosellar or shear   \\ \colrule
    \(-\)            & \(+\)            & \(\bm{[- + +]}\) mesohelical \\ \colrule
    \(0\)            & \(-\)            & neutral mesosellar \\ \colrule
    \(0\)            & \(0\)            & pure reflection \\ \colrule
    \(0\)            & \(+\)            & neutral mesohelical \\ \colrule
    \(+\)            & \(-\)            & \(\bm{[- - +]}\) mesosellar \\ \colrule
    \(+\)            & \(0\)            & \(\bm{[- - +]}\) shear or mesosellar \\ \colrule
    \(+\)            & \(+\)            & \(\bm{[- - +]}\) mesohelical \\[12pt]
  \end{tabular}
\end{subfigure}
\caption{Mesochronic classification based on signs of \(\Sigma\) and \(\Delta\) (see Theorem~\protect\ref{thm:mesochronic-classification}). The classes with \(\Sigma  \not = 0\) are mesohyperbolic. The class for \(\Sigma = \pm \infty\) is the neutral (non-mesohyperbolic) saddle.}\label{tab:mesochronic-classes}
\end{table}
The proof of the theorem will be the result of two lemmas:
\begin{compactitem}
\item Lemma~\ref{lem:mesochronic-classification-timeT} defines \(\Sigma\) and
\(\Delta\) using coefficients of the characteristic polynomial of the
time-\(T\) map Jacobian matrix \(\jac\psi_{T}(x)\).
\item Lemma~\ref{lem:timeT-to-mhjac} establishes relations between
characteristic coefficients of \(\jac\psi_{T}(x)\) and characteristic
coefficients of \(\jac \tilde{f}_{T}(x)\).
\end{compactitem}

\begin{lemma}\label{lem:mesochronic-classification-timeT}
  Let \(\psi_{T}(x)\) be the time-\(T\) map at a point \(x \in X(t_{0}) \subset \R^{3}\) over the time interval \([t_{0},t_{0}+T]\) and
  \(
  P_{\psi}(\mu) := \mu^{3} - t_{\psi} \mu^{2} + m_{\psi} \mu - d_{\psi}
  \), the characteristic polynomial of the Jacobian matrix \(\jac\psi_{T}(x)\).
  Define
  \begin{equation}
    \begin{aligned}
      \Sigma =&\ \frac{t_{\psi} - m_{\psi}}{t_{\psi} + m_{\psi} + 2}, \\
      \Delta =&\ 4(m_{\psi}^{3} + t_{\psi}^{3}) -
      m_{\psi}^{2}t_{\psi}^{2} - 18m_{\psi}t_{\psi} + 27.
    \end{aligned}
    \label{eq:sigma-delta-timet}
  \end{equation}
  If \(m_{\psi} + t_{\psi} + 2 = 0\), then \(x\) is non-mesohyperbolic of
  neutral-saddle type.

  If \(\Sigma\) is finite, then the point \(x\) is classified into mesochronic classes
  defined by Table~\ref{tab:mesochronic-classes}. Mesohyperbolic classes are
  those for which \(\Sigma\) is finite and \(\Sigma \not = 0\), i.e.,
  \begin{equation*}
    t_{\psi} - m_{\psi} \not = 0 \quad \text{and} \quad t_{\psi} + m_{\psi} + 2 \not = 0.
    \label{eq:mesohyperbolicity-timeT}
  \end{equation*}
\end{lemma}

\begin{remark}\label{rem:mf-interpretation}
  The trace \(t_{\psi}\) and the determinant \(d_{\psi}\) of the Jacobian matrix
  \(\jac \psi_{T}\) are commonly encountered in matrix analysis due to their
  simple relationships to eigenvalues. To interpret the cofactor trace \(m_{\psi}\)
  we can expand it as:
  \begin{align*}
    m_{\psi} &= \sum_{j\not = k}\mu_{j}^{\ast}\mu_{k}^{\ast}
               = \mu_{1}^{\ast}\mu^{\ast}_{2}\mu^{\ast}_{3}\sum_{i=1}^{3}\frac{1}{\mu_{i}^{\ast}} \\
             &= \det \jac \psi_{T} \cdot \tr \left[{(\jac \psi_{T})}^{-1}\right].
  \end{align*}
  Under an incompressible flow (\(\det \jac \psi_{T} \equiv 1\)) the cofactor trace \(m_{\psi}\) is the trace of the inverse of the time-\(T\) map Jacobian, as discussed after \eqref{eq:cofactortrace}:
  \begin{equation*}
      m_{\psi} = t_{\psi^{-1}}.
  \end{equation*}

  Additionally, we can clarify the meaning of
  \(\Sigma\). Rewrite~\eqref{eq:sigma-delta-timet} as
  \begin{equation*}
    \Sigma = 1 - 2{\left(\frac{t_{\psi} + 1}{m_{\psi}+1} + 1\right)}^{-1}.
  \end{equation*}
  Then through simple algebraic manipulation it follows that
  \begin{equation}\label{eq:sign-of-sigma}
      \sgn \Sigma = \sgn \left( \frac{t_{\psi} + 1}{m_{\psi} + 1} - 1 \right)
      = \sgn( t_{\psi} - m_{\psi}).
  \end{equation}
  Using \(m_{\psi} = t_{\psi^{-1}}\) and rewriting the expressions for \(t_{\psi}\) using eigenvalues \(\mu_{i}\) of the flow map Jacobian, we obtain that
  \begin{equation}\label{eq:sgn-sigma-as-evals}
    \sgn \Sigma \equiv \sgn \left[
      (\mu_{1} - \mu_{1}^{-1})
      + (\mu_{2} - \mu_{2}^{-1}) - (\mu_{1}\mu_{2} - \mu_{1}^{-1}\mu_{2}^{-1}) \right].
  \end{equation}
\end{remark}

\begin{proof}[Proof of Lemma~\ref{lem:mesochronic-classification-timeT}]
  The mesochronic class of \(x\) is determined by eigenvalues of the Jacobian
  \(\jac \psi_{T}(x)\), whose characteristic polynomial is
  \( P_{\psi}(\mu) = \mu^{3} - t_{\psi} \mu^{2} + m_{\psi} \mu - d_{\psi},
  \) with determinant \(d_{\psi}\), trace \(t_{\psi}\), and cofactor trace \(m_{\psi}\)
  of \(\jac \psi_{T}\) (see Appendix~\ref{sec:cubic-poly-mat}).

The mesochronic class is determined by relative locations of zeros
\(P_{\psi}(\mu^{\ast}) = 0\) in reference to the unit circle and to the
horizontal axis. As mentioned, converting the reference unit circle to a
vertical axes simplifies the criteria, as they can then be computed solely
from signs of eigenvalues. To this end, we introduce the conformal map
\(\Gamma(s) = \frac{1+s}{1-s}\) which maps the left half-plane in \(\C\) to
the inside of the unit circle in \(\C\), while preserving the upper
half-plane. It follows that the location of zeros of the composite function
\(P_{\psi} \circ \Gamma\) with respect to axes is the same as the location of
zeros of \(P_{\psi}\) with respect to horizontal axes and the unit
circle. Note that the inverse \(\Gamma^{-1}(\mu) = \frac{\mu-1}{\mu+1}\) has a
pole at \(\mu^{\ast}=-1\), so no finite zeros in the \(s\)-plane will
correspond to the zero of \(P_{\psi}\) at \(\mu^{\ast}=-1\). For this reason,
we will separately treat the case when \(\jac \psi_{T}\) has an eigenvalue at
\(-1\).

  Assuming \(P_{\psi}(-1) \not = 0\), the composite function is
  \begin{equation}
    \label{eq:composite-characteristic}
    [P_{\psi} \circ \Gamma] (s) = \frac{n_{3} s^{3} + n_{2} s^{2} + n_{1} s + n_{0}}{{(s-1)}^{3}},
  \end{equation}
  with coefficients
  \begin{equation*}
    \begin{aligned}
      n_{3} &= -1 - t_{\psi} - m_{\psi} - \phantom{3}d_{\psi} =  -t_{\psi} - m_{\psi}-2\\
      n_{2} &= -3 - t_{\psi} + m_{\psi} + 3 d_{\psi} = -t_{\psi} + m_{\psi}\\
      n_{1} &= -3 + t_{\psi} + m_{\psi} - 3 d_{\psi} = \phantom{-}t_{\psi} + m_{\psi} - 6\\
      n_{0} &= -1 + t_{\psi} - m_{\psi} + \phantom{3}d_{\psi} = \phantom{-}t_{\psi} - m_{\psi},
    \end{aligned}
  \end{equation*}
  where the second equalities are obtained by
  incompressibility
  \(d_{\psi} = \prod_{i=1}^{3} \mu_{i}^{\ast} = 1\).

  A point \(x\) is non-hyperbolic whenever one of the roots of the
  characteristic polynomial \(P_{\psi}(\mu)\) of \(\jac \psi_{T}(x)\) is on
  the unit circle. If \(P_{\psi}(-1) \not = 0\), this condition implies that a
  purely imaginary number \(ia\), \(a\in\R\) is a zero of \(P_{\psi} \circ
  \Gamma\). Substituting into the numerator of \(P_{\psi}\circ\Gamma\), we
  obtain the condition \(n_{1}n_{2} - n_{0}n_{3} \not = 0\) for
  hyperbolicity. This condition translates to a condition on the spectral
  coefficients of \(\jac \psi_{T}\): \( t_{\psi} - m_{\psi} \not = 0 \).  The
  other case for non-hyperbolicity is \(P_{\psi}(-1) \not = 0\), which by
  direct substitution into \(P_{\psi}\) translates to \( t_{\psi} + m_{\psi} +
  2 \not = 0 \).

  When \(P_{\psi}(-1) \not = 0\), we can proceed to a finer classification of
  the location of the roots of \(P_{\psi}\).  Due to incompressibility, there
  either have to be two zeros of the polynomial \(P_{\psi}\) inside a circle
  and one outside, or vice-versa. Under a conformal transformation \(\Gamma\)
  this condition translates into two zeros of \(P_{\psi} \circ \Gamma\) having
  matching signs, while the third is of the opposite sign. It follows that the
  product of zeros \(s^{\ast}_{i}\) of \(P_{\psi} \circ \Gamma\) is positive
  when two zeros are negative, corresponding to two directions of contraction
  under the flow map, while it is negative when there is a single contracting
  and two expanding directions. By using the zeros \(s_{i}^{\ast}\) of \(P_{\psi}
  \circ \Gamma\) to factorize its numerator \( \sum_{i=0}^{3} n_{i}s^{i} =
  n_{3}\prod_{i=1}^{3}(s - s_{i}^{\ast}) \) and equating the zeroth-order
  coefficients one obtains \(-n_{3} \prod_{i=1}^{3}s_{i}^{\ast} = n_{0}\). We
  therefore define an indicator \(\Sigma\) as
  \[
  \Sigma := s_{1}^{\ast}s_{2}^{\ast}s_{3}^{\ast} = -\frac{n_{0}}{n_{3}} = \frac{t_{\psi} - m_{\psi}}{t_{\psi} + m_{\psi} + 2}.
  \]
  By this argument, \(\Sigma \not = 0\) implies hyperbolicity: \(\Sigma > 0\)
  implies two directions of contraction, one of expansion, while \(\Sigma <
  0\) implies one direction of contraction, two of expansion. When \(\Sigma =
  0\) one of the roots of \(P_{\psi}\) lies on the unit circle, which we term
  ``neutral'', while the other two directions are contracting and expanding.

  The presence of rotation is indicated by the complexity of zeros of \(P_{\psi} \circ \Gamma\), which is determined by the discriminant of the numerator,
  \begin{equation}
  \mathcal{D}_{\threed{}} = -64 \big( 4(m_{\psi}^{3} - t_{\psi}^{3}) - m_{\psi}^{2}t_{\psi}^{2} - 18m_{\psi}t_{\psi} + 27 \big).\label{eq:discriminant}
\end{equation}
When \(\mathcal{D}_{\threed{}} > 0\), all zeros are real and distinct, when
\(\mathcal{D}_{\threed{}} < 0\), two zeros are complex-conjugates of each
other, when \(\mathcal{D}_{\threed{}} = 0\), one real zero is repeated, while
the third is distinct. We therefore define our second indicator \(\Delta\) to
be
\begin{equation}
  \Delta := -\mathcal{D}_{\threed{}}/64 = 4(m_{\psi}^{3} + t_{\psi}^{3}) - m_{\psi}^{2}t_{\psi}^{2} -
  18m_{\psi}t_{\psi} + 27.
\label{eq:delta}
\end{equation}
It follows that \(\Delta < 0\) indicates that all eigenvalues of \(\jac
\psi_{T}(x)\) are real, implying that there is no rigid rotation; \(\Delta >
0\) indicates that two eigenvalues are complex, so rigid rotation is
present. When \(\Delta = 0\) two eigenvalues coincide somewhere along the real
line.

Finally, we have to deal with the case when \(P_{\psi}(-1) = 0\) which is not
covered by the \(s\)-complex plane as defined above, as it implies that the
denominator of \(\Sigma\) is zero, i.e., \(t_{\psi} + m_{\psi} + 2 = 0\). If
one zero is \(\mu^{\ast}_{2} = -1\), then the two others have to be
\(\mu_{1}^{\ast} = z\), \(\mu_{3}^{\ast} = -1/z\), for some \(z\in\C\) to
satisfy incompressibility. Therefore the studied point \(x\) is
non-hyperbolic, with signature \([-\ 0\ +]\). Notice that if \(z\) has an
imaginary component, its conjugate is \(\bar z = -1/z\), which cannot be,
as \(z \bar z \geq 0\), therefore one zero at \(-1\) implies that the other
two zeros are necessarily real. It follows that \(P_{\psi}(-1) = 0\) indicates
a neutral saddle case, even if \(\Sigma\) cannot be evaluated.
\end{proof}

We now relate the characteristic polynomials of \(\jac \psi_{T}\) and \(\jac \tilde{f}_{T}\).
\begin{lemma}\label{lem:timeT-to-mhjac} The characteristic polynomial of the Jacobian matrix \(\jac \psi_{T}\) and the characteristic polynomial of the mesochronic Jacobian \(\jac \tilde{f}_{T}\) are given by the expressions
  \begin{equation}
    \begin{aligned}
      P_{\psi}(\mu) &= \mu^{3} - t_{\psi} \mu^{2} + m_{\psi} \mu - d_{\psi}, \\
      P_{\tilde{f}}(\lambda) &= \lambda^{3} - t_{\tilde{f}} \lambda^{2} + m_{\tilde{f}} \lambda - d_{\tilde{f}}.
    \end{aligned}\label{eq:char-polys}
  \end{equation}
  The coefficients are linked by the expressions
  \begin{equation}
    \begin{aligned}
      t_{\psi} &= 3 + T t_{\tilde{f}}\\
      m_{\psi} &= 3 + 2T t_{\tilde{f}} + T^{2}m_{\tilde{f}}\\
      d_{\psi} &= 1 + T t_{\tilde{f}} + T^{2}m_{\tilde{f}} + T^{3} d_{\tilde{f}},
    \end{aligned}\label{eq:spectral-timeT-to-mesochronic}
  \end{equation}
  where \(T\) is the length of the averaging period in \(\tilde{f}_{T}\).

  Moreover, the incompressibility condition \(d_{\psi} \equiv 1\) imposes the relation
  \begin{equation}\label{eq:incompressibility-mh}
    t_{\tilde{f}} + T m_{\tilde{f}} + T^{2} d_{\tilde{f}} \equiv 0
  \end{equation}
  on the coefficients of \(P_{\tilde{f}}\)  for \(T \not = 0\).
\end{lemma}
\begin{proof}
  The connection between the two Jacobian matrices \(\jac \psi_{T}\) and \(\jac \tilde{f}_{T}\) is given by relation~\eqref{eq:timeT-mh-j-link}: \(\jac \psi_{T} = \id + T \jac \tilde{f}_{T}\). The characteristic polynomial \(P_{\psi}\) of \(\jac \psi_{T}\) can be re-written using the characteristic polynomial \(P_{\tilde{f}}\) of \(\jac \tilde{f}_{T}\)
  \begin{align*}
    P_{\psi}(\mu) &=\ \det(\mu \id - \jac  \psi_{T}) =\ T^{3} \det\left( \frac{\mu - 1}{T} \id - \jac \tilde{f}_{T} \right)\\
                  &= T^{3} P_{\tilde{f}}\left( \frac{\mu - 1}{T} \right).
  \end{align*}
  Using the notation for coefficients of \(P_{\tilde{f}}\) from~\eqref{eq:char-polys}, we can expand the last line to obtain
  \begin{equation}
    \begin{aligned}
      P_{\psi}(\mu) &= \mu^{3} - (3 + T t_{\tilde{f}}) \mu^{2}\\
      &+ (3 + 2 T t_{\tilde{f}} + T^{2} m_{\tilde{f}} ) \mu \\
      &- (1 + T t_{\tilde{f}} + T^{2} m_{\tilde{f}} + T^{3} d_{\tilde{f}}).
    \end{aligned}
  \end{equation}
  By comparing coefficients with the general expression for \(P_{\psi}\), the statement of the theorem follows. The incompressibility condition from the statement of the theorem
  is a consequence of substituting \(d_{\psi} \equiv 1\) for any \(T > 0\).
\end{proof}

We now combine the preceding Lemmas to give the proof of Theorem~\ref{thm:mesochronic-classification}.
\begin{proof}[Proof of Theorem~\ref{thm:mesochronic-classification}.]

  We start with the expression for the \(P_{\psi}(-1) \not= 0\) condition, which was stated as \(m_{\psi} + t_{\psi} + 2 \not = 0\). Using~\eqref{eq:spectral-timeT-to-mesochronic} we can rewrite the condition as
  \[
  8 + 3 t_{\tilde{f}}T + m_{\tilde{f}} T^{2} \not = 0.
  \]
  Since \(t_{\tilde{f}}\), \(m_{\tilde{f}}\), and \(d_{\tilde{f}}\) are related through the incompressibility constraint \(t_{\tilde{f}} + m_{\tilde{f}}T + d_{\tilde{f}}T^{2} \equiv 0\) derived in~\eqref{eq:incompressibility-mh}, we can formulate the condition in two alternative ways:
  \begin{align*}
    8 + 2 t_{\tilde{f}}T - d_{\tilde{f}} T^{3} &\not = 0 \\
    8 - 2 m_{\tilde{f}}T^{2} - 3d_{\tilde{f}} T^{3} &\not = 0.
  \end{align*}

  The other mesohyperbolicity condition \(t_{\psi}-m_{\psi} \not = 0\) translates into either
  \begin{equation*}
    8 d_{\tilde{f}} T^{3} \not = 0 \quad \text{or} \quad
    8T(t_{\tilde{f}} + m_{\tilde{f}}T) \not = 0.
  \end{equation*}

  Reformulating the expressions for \(\Sigma\) and \(\Delta\) in Lemma~\ref{lem:mesochronic-classification-timeT}
  \begin{align*}
    \Sigma &= \frac{t_{\psi} - m_{\psi}}{t_{\psi} + m_{\psi} + 2} =  1 - 2\frac{m_{\psi}t_{\psi}}{t_{\psi} + m_{\psi} + 2}\\
    \Delta &= 4(m_{\psi}^{3} + t_{\psi}^{3}) - m_{\psi}^{2}t_{\psi}^{2}
             - 18m_{\psi}t_{\psi} + 27,
  \end{align*}
  using \(d_{\tilde{f}}\), \(m_{\tilde{f}}\) and \(t_{\tilde{f}}\) through~\eqref{eq:spectral-timeT-to-mesochronic} we obtain
  \begin{align*}
    \Sigma =&\ -T \frac{t_{\tilde{f}} + T m_{\tilde{f}}}{8 + 3T t_{\tilde{f}} + T^{2}m_{\tilde{f}}} \\
    \Delta =&\
              (4 m_{\tilde{f}}^{3} - m_{\tilde{f}}^{2}t_{\tilde{f}}^{2}) T^{6} +
              (18m_{\tilde{f}}^{2}t_{\tilde{f}} - 4m_{\tilde{f}}t_{\tilde{f}}^{3}) T^{5} \\
            & +
              (27 m_{\tilde{f}}^{2} + 18 m_{\tilde{f}}t_{\tilde{f}}^{2} - 4t_{\tilde{f}}^{4}) T^{4}  \\
            & +
              54 m_{\tilde{f}} t_{\tilde{f}}T^{3} +
              27 t_{\tilde{f}}^{2} T^{2}.
  \end{align*}

  Since \(t_{\tilde{f}}\), \(m_{\tilde{f}}\), and \(d_{\tilde{f}}\) are related through the incompressibility constraint~\eqref{eq:incompressibility-mh}, we can re-formulate the expressions using either \(t_{\tilde{f}}\) or \(m_{\tilde{f}}\)
  \begin{align*}
    \Sigma =&\ \frac{ d_{\tilde{f}} T^{3}}{8 + 2 t_{\tilde{f}}T - d_{\tilde{f}}T^{3}}
              = \frac{d_{\tilde{f}}T^{3}}{8-2m_{\tilde{f}}T^{2}-3 d_{\tilde{f}}T^{3}},\\
    \Delta =&\ -d_{\tilde{f}}^{3}T^{9} - d_{\tilde{f}}^{2}t_{\tilde{f}}T^{8} +
              6d_{\tilde{f}}^{2}t_{\tilde{f}} T^{7} \\
            & + (27 d_{\tilde{f}} + 2t_{\tilde{f}}^{3})d_{\tilde{f}}T^{6}
              + 6d_{\tilde{f}}t_{\tilde{f}}^{2} T^{5}  - t_{\tilde{f}}^{4} T^{4} - 4t_{\tilde{f}}^{3}T^{3}\\
    =&\ -4 d_{\tilde{f}}^{4} T^{12} - 12 d_{\tilde{f}}^{3}m_{\tilde{f}} T^{11} \\
            &- 13 d_{\tilde{f}}^{2}m_{\tilde{f}}^{2}T^{10} - 6d_{\tilde{f}}m_{\tilde{f}}^{3}T^{9}
              + (18 d_{\tilde{f}}^{2}m_{\tilde{f}}
                                                        - m_{\tilde{f}}^{4})T^{8} \\ &+ 18 d_{\tilde{f}}m_{\tilde{f}}^{2} T^{7} + (27 d_{\tilde{f}}^{2} + 4 m_{\tilde{f}}^{3})T^{6}.
  \end{align*}
  To emphasize the connection to the OWC expressions (Theorem~\ref{thm:3D-okubo-weiss}), we will use the \(d_{\tilde{f}}\) and \(m_{\tilde{f}}\) versions of the formulas.

  The statement of the proof is equivalent to Lemma~\ref{lem:mesochronic-classification-timeT} where the criteria are expressed in terms of the spectral coefficients of \(\jac \tilde{f}_{T}\) instead of \(\jac \psi_{T}\).
\end{proof}

In summary, to identify the mesochronic class of a point \(x\), take the following steps:
\begin{compactenum}
\item compute the Jacobian \(\jac \tilde{f}_{T}(x)\) (details in  Appendices~\ref{sec:ode-mh-jacobian} and~\ref{sec:implementation}),
\item evaluate \(\Delta\) and \(\Sigma\) using~\eqref{eq:sigmadelta},
\item use Table~\ref{tab:mesochronic-classes} to identify the mesochronic class.
\end{compactenum}

\section{Limits and boundary cases}
\label{sec:limits-boundary-cases}

When, respectively, \(T\to 0\) and \(T \to \infty\), we show that mesochronic
classification limits to classical \owc{} and Lyapunov exponent
analyses. Additionally, if the dynamics evolves on an invariant plane in
\threed{} space, certain \threed{} mesochronic classes have their counterparts
in the \twod{} mesochronic classification~\cite{Mezic2010}.

\subsection{Instantaneous limit: The
  \owc{} criterion}\label{sec:limit-to-zero}

The relation between flow map and mesochronic velocity~\eqref{eq:timeTmap} can
be rewritten as $\tilde{f}_T(x) = \frac{\phi(t_0+T,t_0,x) - x}{T}$. As a
consequence, the averaged field \(\tilde{f}_{T}\) tends to $f$ pointwise as
\(T \to 0^{+}\) and \(\jac \tilde{f} \to \jac f\). By continuity of
eigenvalues with respect to matrix elements, \(d_{\tilde{f}} \rightarrow
d_{f}\), \(m_{\tilde{f}} \rightarrow m_{f}\), and \(t_{\tilde{f}} \rightarrow
t_{f}\) as $T \to 0^{+}$. As all three of these quantities are finite, the
mesochronic incompressibility criterion~\eqref{eq:incompressibility-mh} is
trivially satisfied in the limit. Additionally, due to incompressibility of
the vector field, it holds that \(t_{\tilde{f}} \xrightarrow{T \to 0} 0\).

\begin{theorem}\label{thm:owc-to-mh}
  Suppose that at the point \((t_{0},x)\) the differential
  equation~\eqref{eq:dynamical-system} is instantaneously hyperbolic by the
  \owc{} (OWC) criterion. Then, there exists \(T_{\min} > 0\) such that the point
  \((t_{0},x)\) is also mesohyperbolic with respect to all time intervals
  \([t_{0},t_{0}+T]\) for which \(T < T_{\min}\).
\end{theorem}
\begin{proof}

  According to Theorem~\ref{thm:3D-okubo-weiss}, we obtain that
  \(d_{f}\not=0\). Since the maps \(T \mapsto \jac \tilde{f}_{T}\), \(T
  \mapsto d_{\tilde{f}_{T}}\), and \(T \mapsto m_{\tilde{f}_{T}}\) are
  continuous, the instantaneous hyperbolicity will imply mesohyperbolicity for
  some small \(T_{\min} > 0\). The continuity means that there exists
  \(T_{\min} > 0\) on which continuity of the maps $T \mapsto
  d_{\tilde{f}}(T)$ and $T \mapsto 3t^3d_{\tilde{f}_{T}}+2 T^2
  m_{\tilde{f}_{T}}-8$ implies
  \[
  d_{\tilde{f}_{T}}\not=0,\quad  3T^3d_{\tilde{f}_{T}}+2T^2m_{\tilde{f}_{T}}-8\not=0,
  \]
  for \(T \in [0,T_{\min}]\).  By virtue of
  Theorem~\ref{thm:mesochronic-classification}, the point $(t_{0},x)$
  is also mesohyperbolic with respect to $[0,T]$ and the proof is complete.
\end{proof}

As a consequence, the signs of the indicators \(\Delta\) and
\(\Sigma\)~\eqref{eq:sigmadelta} have the following limits
\begin{equation*}
  \begin{aligned}
    \sgn \Sigma &\xrightarrow{T \to 0} \sgn d_{f},\\
    \sgn \Delta &\xrightarrow{T \to 0} \sgn (27 d_{f}^{2} + 4 m_{f}^{3}).
  \end{aligned}
\end{equation*}
By comparing mesochronic classification criteria
(Table~\ref{tab:mesochronic-classes}) in the limit \(T \to 0\) with the
instantaneous OWC criterion (Theorem~\ref{thm:3D-okubo-weiss}), we can conclude
that mesochronic classification reduces to OWC classification in the \(T\to
0\) limit. Put differently, mesochronic classes generalize OWC classes to
time intervals $[t_0,t_0 +T]$ with $T > 0$.

\subsection{Asymptotic limits}\label{sec:limit-to-infty}

For autonomous dynamical systems defined over compact domains, asymptotic
rates of deformation are defined by Lyapunov exponents. The finite-time
analogs, termed \emph{\ftle{}} (FTLE)~\cite{Shadden2005,Haller2001} are defined using
the polar decomposition of the Jacobian matrix of the time-\(T\) map
\(\psi_{T}\). Let \(\jac\psi_{T} = R \abs{\jac\psi_{T}}\) be the polar
decomposition of \(\jac\psi_{T}\) where \(\abs{\jac\psi_{T}}^{2} =
\jac\psi_{T}^{\top} \jac\psi_{T} \). Eigenvalues of \(\abs{\jac\psi_{T}}\) are
non-negative, singular values of \(\jac\psi_{T}\), so we can define
\ftle{} \(\sigma_{i}\) to be the exponential
growth/decay rates of the singular values, i.e., we represent the singular
values of \(\jac\psi_{T}\) as \(e^{\sigma_{i} T}\).

It is a well-known fact in matrix analysis that when a matrix is normal, i.e.,
unitarily diagonalizable, its singular values are equal to absolute values of
its eigenvalues. In the language of this paper, this means that when \(\jac
\psi_{T}\) is normal, positions of the eigenvalues \(\mu_{i}\) of \(\jac
\psi_{T}\) in reference to the unit circle are determined by the signs of
the \ftle{} \(\sigma_{i}\). It follows that \(\Sigma > 0\) (Table~\ref{tab:mesochronic-classes}) implies that both two eigenvalues of \(\jac \psi_{T}\) are outside of the unit circle \emph{and} that two \ftle{} are positive, when \(\jac \psi_{T}\) is normal.

Unfortunately, \(\jac \psi_{T}\) is not generally normal for any finite \(T\),
i.e., its eigenvectors are not orthonormal. However,
Ref.~\cite{Goldhirsch1987} shows, although non-rigorously, that for
smooth ergodic systems which have distinct Lyapunov exponents both real parts
of eigenvalues and singular values of \(\jac\psi_{T}\) can be written as
\(e^{\sigma_{i}T}\) as \(T \to \infty\). As the sign of \(\Sigma\) indicates
the number of eigenvalues of \(\jac\psi_{T}\) outside the unit circle, which, in
turn, is determined by real parts of logarithms of those eigenvalues, we
conclude that, when \(T \to \infty\), the sign of \(\Sigma\) indicates whether
two or one Lyapunov exponents are positive, assuming that the conjecture in
Ref.~\cite{Goldhirsch1987} holds.

\subsection{Recovering the \twod{} mesochronic deformation criterion}\label{sec:degeneracy}

The supplement to Ref.~\cite{Mezic2010} presented a derivation of the
criteria for mesohyperbolicity for \twod{} (planar)
differential equations. Planar differential equations can be trivially
embedded into a \threed{} state space by adding a third state with trivial
(zero) dynamics. We use such an embedding to demonstrate how the \threed{}
mesochronic criteria (Theorem~\ref{thm:mesochronic-classification}) specialize to
the \twod{} criterion.

Let \(g : I \times \R^{2} \to \R^{2}\) be a \(C^{2}\) incompressible (\(\nabla
\cdot g \equiv 0\)) vector field, with mesochronic Jacobian \(\jac
\tilde{g}_{T}(x)\) for \(x \in \R^{2}\), with the spectrum
\(\sigma_{\tilde{g}} = \{\lambda_{1}, \lambda_{2}\}\), and trace and
determinant \(t_{\tilde{g}} = \lambda_{1} + \lambda_{2}\), \(d_{\tilde{g}} =
\lambda_{1}\lambda_{2}\). The incompressibility
implies
\begin{equation}
  t_{\tilde{g}} + T d_{\tilde{g}} \equiv 0,\label{eq:2D-incompressibility}
\end{equation}
with eigenvalues then given by
\begin{equation}
  \lambda_{1,2} = -\frac{T}{2} d_{\tilde{g}} \pm \frac{1}{2} \sqrt{(T^{2} d_{\tilde{g}} - 4) d_{\tilde{g}}}.\label{eq:2D-eigenvalues}
\end{equation}
Ref.~\cite{Mezic2010} studied only \(d_{\tilde{g}} \not = 0\), noting
that \(d_{\tilde{g}} = 0\) results in \(\lambda_{1} = \lambda_{2} = 0\).  The
time-\(T\) map of the velocity field \(\dot x = g(x)\) is hyperbolic at
\(x_{0}\) if it preserves two distinct real spatial axes, which is analogous
to the definition in Section~\ref{sec:diffeomorphism}. Consequently,
Definition~\ref{def:mesohyperbolicity} retains its meaning for the \twod{}
case. The discriminant \(\mathcal{D}_{\twod{}} := (T^{2} d_{\tilde{g}} - 4) d_{\tilde{g}}\) (cf. \threed{} discriminant~\eqref{eq:discriminant}) indicates mesohyperbolicity if \(\mathcal{D}_{\twod{}} > 0\) and mesoellipticity otherwise.

We embed the vector field \(g\) in the \threed{} state space by
defining a \threed{} vector field for \(z \in \R^{3}\) as
\begin{equation}
  \begin{aligned}
    f(t,z) &:=
    \left[\begin{smallmatrix}
        g(t,z_{1,2}) \\ 0
      \end{smallmatrix}\right],  \\
    \jac \tilde{f}(z) &=
    \left[\begin{smallmatrix}
        \jac  \tilde{g}(z_{1,2}) & 0 \\
        0 & 0
      \end{smallmatrix}\right],
  \end{aligned}
\end{equation}
where we take \(z_{1,2} := {[z_{1},z_{2}]}^{\rT}\) to be the first two
components of the vector \(z = {[z_{1},z_{2},z_{3}]}^{\rT}\).  Eigenvalues of
the Jacobian of the \threed{} averaged field are \(\sigma_{\tilde{f}} =\allowbreak
\{\lambda_{1},\allowbreak \lambda_{2}, \allowbreak \lambda_{3}\}\), with \(\lambda_{3} \equiv
0\). The spectral quantities \(t_{\tilde{f}}\), \(m_{\tilde{f}}\), and
\(d_{\tilde{f}}\) then reduce to analogous quantities of \(\jac
\tilde{g}_{T}\)
\begin{equation}
  \begin{aligned}
    t_{\tilde{f}} &= \textstyle{\sum_{i=1}^{3}} \lambda_{i} = t_{\tilde{g}},\\
    m_{\tilde{f}} &= \textstyle{\sum_{i=1}^{3} \prod_{k\not =i}} \lambda_{k} = \lambda_{1}\lambda_{2} = d_{\tilde{g}},\\
    d_{\tilde{f}} &= 0.
  \end{aligned}
  \label{eq:2D-spectral-quant}
\end{equation}
As a consequence, the \threed{} incompressibility
condition~\eqref{eq:incompressibility-mh} reduces to the \twod{}
incompressibility condition~\eqref{eq:2D-incompressibility}.

The \threed{} conditions for mesohyperbolicity~\eqref{eq:mesohyperbolicity}
reduce to
\[
d_{\tilde{f}} \not = 0 \quad\text{and}\quad 4 - d_{\tilde{g}}T^{2} \not = 0.
\]
Since \(d_{\tilde{f}} \equiv 0\) as noted above, it follows that the flow is
not \threed{}-mesohyperbolic, as it is expected, as the third coordinate is
always preserved due to the construction of the \threed{} flow.

The indicators \(\Sigma\) and \(\Delta\) evaluate to
\[
\Sigma = 0, \quad \Delta = -m_{\tilde{f}}^{3}T^{6}(m_{\tilde{f}}T - 4) = -d_{\tilde{g}}^{2}T^{6} \mathcal{D}_{\twod{}}.
\]
Therefore, the sign of the \(\Delta\) indicator is determined by the
\twod{} mesohyperbolicity criterion. If \(\Delta > 0\),
\(\mathcal{D}_{\twod{}} < 0\) and the two eigenvalues are non-real and lie on
the unit circle due to incompressibility constraints. If \(\Delta < 0\),
\(\mathcal{D}_{\twod{}} > 0\) and the eigenvalues are real, one is contracting
and the other expanding.

For planar dynamics, \(\mathcal{D}_{\twod{}} = 0\) implies that the dynamics are
either a pure shear, when the geometric multiplicity of the eigenvalue is
\(1\), or trivial, i.e., \(\psi_{T}(x) \equiv x\).

Incompressibility in the planar case is more restrictive than in \threed{},
yielding only two structurally stable cases: mesohyperbolic \(\lambda_{i} \in
\R\) and mesoelliptic \(\lambda_{i} \in \C\), which intersect at the pure
reflection/shear case \(\lambda_{1} = \lambda_{2} = 1\). The derivation of the
mesohyperbolicity criterion for the \twod{} case therefore relied only on
detection of real vs.\ complex eigenvalues, which is the reason why
\(\mathcal{D}_{\twod{}}\) in~\eqref{eq:2D-eigenvalues} is taken as the sole
\twod{} mesohyperbolicity criterion, without resorting to a more complicated
calculation.

\section{Examples}
\label{sec:computation}

\subsection{Linear time-invariant velocity fields}
\label{sec:ex-linear}

 \begin{table}  \renewcommand{\arraystretch}{2} \centering\footnotesize
\begin{subtable}[t]{\linewidth}\centering
\( \psi_{T} =
\left(\begin{smallmatrix} \cos \omega T & -\sin \omega T & \phantom{0}\\ \sin
\omega T &\phantom{-}\cos \omega T & \phantom{0} \\ \phantom{0} & \phantom{0}
& 1
    \end{smallmatrix}\right)\exp \left(\begin{smallmatrix} \lambda_{1}T &
\phantom{\lambda_{2}} & \phantom{-\lambda_{1} - \lambda_{2}} \\
\phantom{\lambda_{1}} & \lambda_{2}T & \phantom{-\lambda_{1} - \lambda_{2}}\\
\phantom{\lambda_{1}} & \phantom{\lambda_{2}} & -(\lambda_{1} + \lambda_{2})T
    \end{smallmatrix}\right) \)
\caption{Form of the considered LTI systems}
\end{subtable}\\
\begin{subtable}[t]{\linewidth}\centering
    \begin{tabular}{ m{20pt} c c c } \(\omega\) & \(\lambda_{1},\ \lambda_{2}\) &
\(\Sigma(T)\) & \(\Delta(T)\) \\

 \colrule \(0\) & \(\lambda_{1}\cdot\lambda_{2} > 0\) &
\(\begin{aligned} -&\tanh\lambda_1 T/2\\ \times&\tanh \lambda_2 T/2\\
\times&\tanh(\lambda_1+\lambda_{2})T/2 \end{aligned} \) &
\(\begin{aligned}[c]-64 &\sinh^{2}(\lambda_{1}-\lambda_{2})T/2 \times \\
 \phantom{\times} &\sinh^{2}(2\lambda_{1}+\lambda_{2})T/2 \times \\
 \phantom{\times} &\sinh^{2}(\lambda_{1}+2\lambda_{2})T/2\end{aligned}\) \\

 \colrule \(0\) & \(\lambda = \lambda_{1} = \lambda_{2} \not= 0\) &
\(\tanh\lambda T - 2 \tanh \lambda T/2\) &\(0\) \\

 \colrule \(0\) & \(\lambda_{1} = \lambda \not = \lambda_{2} = 0\) & \(0\)&
\(\begin{aligned}-64 & \sinh^{4}( \lambda T/2)\times\\\phantom{\times}&\sinh^{2}(\lambda T)\end{aligned}\) \\

 \colrule \(\not = 0\) & \(\lambda = \lambda_{1} = \lambda_{2} \not = 0\) &
\(\frac{\tanh(\lambda T) (\cos(\omega T) -\cosh(\lambda T)) }{\cosh(\lambda T)
+ \cos(\omega T)}\) & \(
\begin{aligned}16 & \sin^{2}(\omega T)\times \\ \phantom{\times} & {[\cos (\omega T) - \cosh(3
\lambda T)]}^{2}
\end{aligned}\) \\
 \colrule \(\not = 0\) & \(\lambda_{1} = \lambda_{2} = 0\) & \(0\) & \(64
\sin^{4}( \omega T/2)\sin^{2}(\omega T)\)
    \end{tabular}
\caption{Values of \(\Delta\) and \(\Sigma\)}
\end{subtable}\\
\begin{subtable}[t]{\linewidth}\centering
    \begin{tabular}{ c c c c c } \(\omega\) & \(\lambda_{1},\ \lambda_{2}\) &
\(\sgn\Sigma(T)\) & \(\sgn\Delta(T)\) & Mesochronic class \\ \colrule
 \(0\) & \(\lambda_{1}\cdot\lambda_{2} > 0\)& \(
-\sgn\lambda_{1}\) & \(-\) & \parbox[m]{3.5cm}{ \([- + +]\) saddle
(\(\lambda_{1} > 0\)) \\ \([- - +]\) saddle (\(\lambda_{1} < 0\)) } \\
\colrule \(0\) & \(\lambda = \lambda_{1} = \lambda_{2} \not= 0\) &
\(-\sgn \lambda\) & \(0\) & \parbox[m]{3.5cm}{ \([- + +]\) saddle
(\(\lambda_{1} > 0\)) \\ \([- - +]\) saddle (\(\lambda_{1} < 0\)) } \\
\colrule \(0\) & \(\lambda_{1} \not = \lambda_{2} = 0\) & \(0\) & \(-\)
& \parbox[m]{3.5cm}{neutral
saddle\\ \twod{} mesohyperbolic} \\ \colrule
\(\not = 0\) & \(\lambda = \lambda_{1} = \lambda_{2} \not = 0\) & \(-\sgn\lambda\) & \(+\)
& \parbox[m]{3.5cm}{ \([- + +]\) helix (\(\lambda_{1} > 0\)) \\ \([- - +]\)
helix (\(\lambda_{1} < 0\)) } \\ \colrule
\(\not = 0\) &
\(\lambda_{1} = \lambda_{2} = 0\) & \(0\) & \(+\) & \parbox[m]{3.5cm}{neutral
helix\\ \twod{} mesoelliptic} \end{tabular}
\caption{Mesochronic classes}
\end{subtable}\caption{\label{tab:LTI}Mesochronic classes for linear time-invariant (LTI) systems of
the form (a). Values of signs in (c) hold generically, except on a non-dense
set of periods \(T\) where they are zero, as determined by values \(\omega\)
and \(\lambda\). }
  \end{table}

In Table~\ref{tab:LTI} we compute explicitly the values of the indicators \(\Sigma\) and \(\Delta\) for a simple class of linear time-invariant systems whose \(\jac \psi_{T}\) is constant, and given by the polar decomposition:
\begin{equation}
  \label{eq:lti-flow-map}
    \jac \psi_{T} \equiv \left(\begin{smallmatrix}
        \cos \omega T & -\sin \omega T & \phantom{e^0}\\
        \sin \omega T &\phantom{-}\cos \omega T & \phantom{e^0} \\
        \phantom{e^0} & \phantom{0} & 1
      \end{smallmatrix} \right)
    \left(\begin{smallmatrix}
        e^{\lambda_{1}T} & \phantom{e^{\lambda_{1}T}} & \phantom{e^{\lambda_{1}T}} \\
        \phantom{e^{\lambda_{1}T}} & e^{\lambda_{2}T} & \phantom{e^{\lambda_{1}T}}\\
        \phantom{e^{\lambda_{1}T}} & \phantom{e^{\lambda_{1}T}} & e^{\lambda_{3}T}
      \end{smallmatrix}\right).
\end{equation}
In this parametrization, we can independently manipulate rates of strain
\(\lambda_{1,2,3}\) as well as the rate of rotation \(\omega\) present in the
system. As two of the rates have the same sign, unless one of the rates is
zero, we choose to order the directions by setting \(\sgn \lambda_{1} = \sgn
\lambda_{2}\), which means that the third direction is of the opposite sign
\(\lambda_{3} = -\lambda_{1}-\lambda_{2}\), due to incompressibility. While
these systems do not represent a broad range of dynamical systems, we have a
good understanding of their dynamics so it is instructive to see how their
properties are reflected in the mesochronic classification.

First, when all rates \(\lambda_{i}\) are non-zero, all points are mesohyperbolic as \(\Sigma\) is constant and non-zero; presence or absence of rotation determines whether a point is a (mesohyperbolic) saddle (\(\omega=0\)) or a helix (\(\omega \not = 0\)). The signature \(\bm{[- - +]}\) or \(\bm{[- + +]}\) of the saddle is then determined by the sign of the pair \(\lambda_{1,2}\). When the two rates match exactly, \(\lambda_{1} = \lambda_{2}\), it implies \(\Delta = 0\). Since the quantities \(\Delta\) and \(\Sigma\) are functions of the spectrum, they alone are not enough to detect whether associated directions align (shear) or not (saddle). In the case of the systems derived, we know that those two directions correspond to independent eigenvectors, which means that the point is a saddle.

If one of the rates is equal to \(0\), it always implies that \(\Sigma = 0\),
which is classified as one of the neutral \threed{} mesochronic classes. In
that case, \(\Delta\) corresponds to the \twod{} mesohyperbolicity indicator
\(\mathcal{D}_{\twod{}}\), as described in
Section~\ref{sec:degeneracy}. Again, the presence \(\omega\not=0\) or the
absence of rotation \(\omega=0\) is reflected on the sign of \(\Delta\), where
\(\Delta > 0\) corresponds to the former, and \(\Delta < 0\) to the latter
case.

The magnitudes of \(\Delta\) and \(\Sigma\) grow exponentially in most of the
cases; however, notice that in the presence of rotation, a periodic function
multiplies the exponentially-growing magnitude, resulting in \(\Delta(T) = 0\)
periodically. This means that there is a potential for resonance, i.e., if
\(T\) is a multiple of the period of oscillation, dynamics momentarily appears
to be on the boundary behavior between \(\bm{[- + +]}\) and \(\bm{[- - +]}\)
saddle mesohyperbolicity, or even a pure reflection when \(\Delta = \Sigma =
0\). This choice of \(T\) is, of course, highly unlikely without a prior
knowledge of \(\omega\).

In summary, analysis of simple linear systems shows that mesochronic classes correctly reflect our intuition about presence of stretching and rotation in linear, time-invariant flows.

\subsection{\abc{} Flow}
\label{sec:abc-flow}

The \abc{} (ABC) flow~\cite{Dombre1986} is a kinematic model of an incompressible fluid
flow evolving in a three-dimensional periodic domain. Even though the system
of ODEs specifying the ABC flow is simple, it exhibits a variety of different
behaviors and has been used as a test-bed for various computational
algorithms~\cite{Haller2001, Froyland2009, Budisic2012b, Brunton2010}.

The ABC flow evolves on a 3-torus in periodized state variables \((x,y,z) \in {[0,2\pi]}^{3} \cong \T^{3}\). Dynamics depend on parameters \(A,B,C,D \in \R\) and are specified by differential equations
\begin{equation}    \label{eq:abcmodel}
  \begin{aligned}
    \dot x &= A(t) \sin  z + \phantom{(t)}C \cos  y \\
    \dot y &= \phantom{(t)}B \sin  x + A(t) \cos  z  \\
    \dot z &= \phantom{(t)}C \sin  y + \phantom{(t)}B \cos  x,
  \end{aligned}
\end{equation}
where the time-varying parameter \(A(t)\) is given by
\begin{equation*}
  A(t) = A + D\, t \sin t.
\end{equation*}
If \(D = 0\), the equations are autonomous; if, additionally, any other parameter is \(0\), the system is integrable.~\cite{Dombre1986}

The linearization along a solution \(p(t) = (x(t),y(t),z(t))\)
of~\eqref{eq:abcmodel} is given by
\begin{equation}\label{eq:abclinear}
  \dot\xi=\underbrace{
    \left[\begin{smallmatrix}
        0 & -C\sin y(t) & A(t)\cos  z(t) \\
        B\cos x(t) & 0 & -A(t)\sin z(t) \\
        -B\sin  x(t) & C\cos y(t) & 0
      \end{smallmatrix}\right]}_{\jac f(x,y,z)}
  \xi.
\end{equation}
The determinant and the sum of minors are given by the expressions
\begin{equation}
  \begin{aligned}
    d_{f} = \det \jac f &= A(t)BC(
    \begin{aligned}[t]
      & \cos x \cos y \cos z \\ & - \sin x \sin y \sin z)
    \end{aligned} \\
    m_{f} = \tr \cof \jac  f &=
    \begin{aligned}[t]
      & A(t)B \sin x \cos z  \\ + & BC \sin y \cos x \\ + & A(t)C \cos y \sin z.
    \end{aligned}
  \end{aligned} \label{eq:abc-instant}
\end{equation}
These expressions can be used to evaluate the OWC criterion for the instantaneous hyperbolicity according to Theorem~\ref{thm:3D-okubo-weiss}.

\subsubsection{Integrable case}
We briefly discuss the case \(A=D=0 \Rightarrow A(t) \equiv 0\)
analytically. From~\eqref{eq:abcmodel}, we derive that
\begin{equation*}
    \ddot{z} = C \dot y \cos y - B \dot x \sin x = BC \sin x \cos y - BC \sin x \cos y \equiv 0.
\end{equation*}
Thus, for the initial condition \(p(0) = (x_{0},y_{0},z_{0})\), \(z(t)=z_0+(C\sin
y_{0}+ B\cos x_0) t\) for all \(t\in [0,T]\).

Furthermore, if we write \(\sigma := x + y\), \(\delta := x - y\), then the
ABC flow can be rewritten as a decoupled second order system
\begin{equation}
  \label{eq:decoupled-abc}
  \begin{aligned}
    \ddot \sigma &= \phantom{-}BC \cos \sigma \\
    \ddot \delta &= -BC \cos \delta \\
    \ddot z &= 0,
  \end{aligned}
\end{equation}
which is a direct product of two pendulum-like equations. Solutions of the
first two components can be written in terms of integrals of Jacobi elliptic
functions, and it follows that the system is integrable. A similar argument
follows in case when either \(A\), \(B\), or \(C\) are zero, in addition to
\(D = 0\).

\begin{lemma} All points \((x,y,z)\) in the state space of the
  system~\eqref{eq:abcmodel} with \(A(t)\equiv 0\) are non-mesohyperbolic over
  any interval \([t_{0},t_{0}+T]\).
\end{lemma}
\begin{proof}
  When \(A(t) \equiv 0\), the matrix defining the linear system of
  equations~\eqref{eq:abclinear} is block-diagonal, with blocks \(
  \left[\begin{smallmatrix}
      0 & -C\sin y(t)  \\
      B\cos x(t) & 0
    \end{smallmatrix}\right]\) and \(0\) on the diagonal. The fundamental
  matrix of the system is, therefore, also block diagonal, with value \(1\) on
  the diagonal corresponding to the exponential of the block \(0\)
  in~\eqref{eq:abclinear}. Since \(1\) is then in the spectrum of the time-\(T\)
  map Jacobian, \((x,y,z)\) is non-mesohyperbolic for all \(T\).
\end{proof}

\begin{remark}
  If \(A(t) \equiv 0\), the mesochronic class of a point \((x,y,z)\) does not
  depend on the value of \(z\) since the Jacobian matrix
  in~\eqref{eq:abclinear} does not depend on the \(z\)-coordinate of the
  solution around which we linearized.
\end{remark}

\subsubsection{Steady non-integrable case}

The structure of the invariant sets in the state space of the ABC flow for
parameters \(A=\sqrt{3},\,B=\sqrt{2},\,C=1,\,D=0\) is well studied
analytically~\cite{Dombre1986} and numerically~\cite{Budisic2012b}. The state
space contains six interwoven vortices with the space between them filled by
chaotic dynamics (Figure~\ref{fig:abcflow-statespace}). We place a grid of \(400
\times 400\) initial conditions on the \((x,y)\) face of the periodicity cell
and calculate \(t_{\tilde{f}},\,m_{\tilde{f}},\,d_{\tilde{f}}\) for time
intervals of different lengths. Other details about numerics are given in Appendix~\ref{sec:implementation}.
\begin{figure}[!ht]
  \centering
  \begin{subfigure}[t]{0.45\linewidth}\centering
    \includegraphics{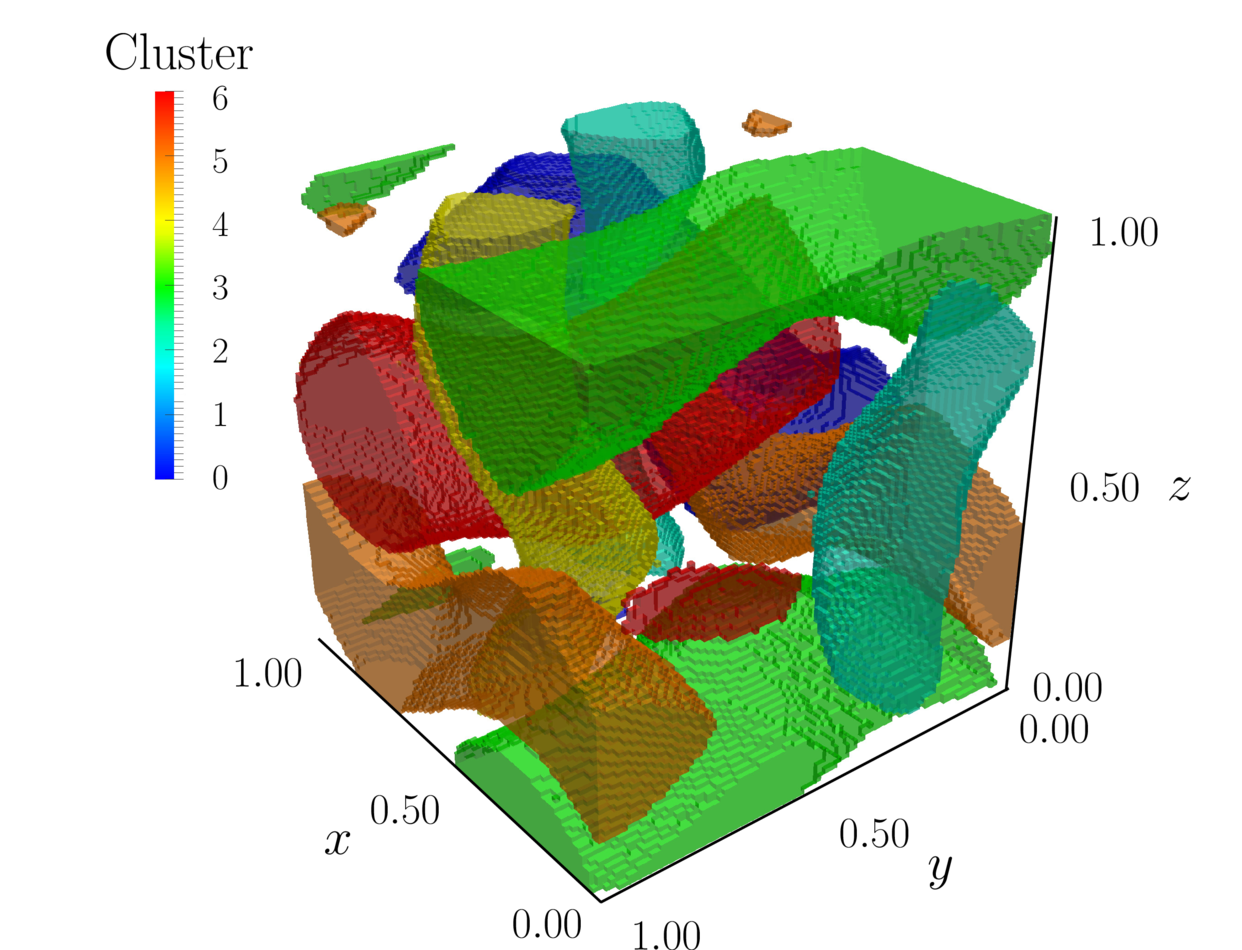}
    \caption{Isometric view (periodicity cube rescaled to unit sides).}
  \end{subfigure}\quad
  \begin{subfigure}[t]{0.45\linewidth}\centering
    \includegraphics{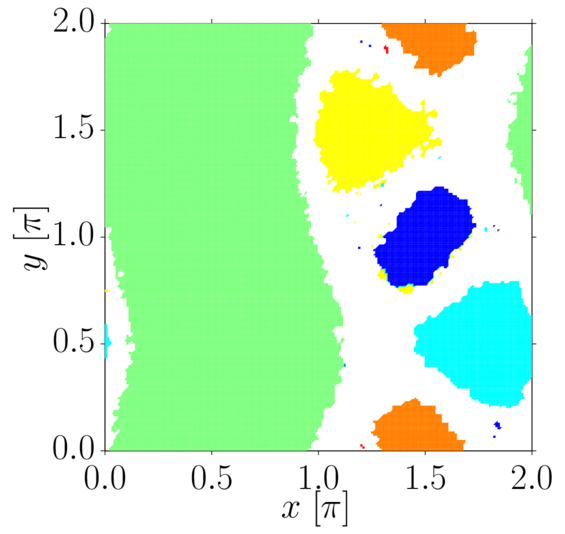}
    \caption{Slice through \(z = 0\).}\label{fig:vortices-slice}
  \end{subfigure}
  \caption{Invariant sets in the state space of the \abc{} flow at
    \(A=\sqrt{3},\,B=\sqrt{2},\,C=1\). Regular vortices are colored, the space
    between them is the chaotic zone. }\label{fig:abcflow-statespace}
\end{figure}

To give a sense of time scales involved in the system,
Figure~\ref{fig:abc-steady} shows several trajectories (pathlines) within a
single vortex, simulated for various durations \(T\). Trajectories inside
vortices take approximately \(T = 3\) to cross one periodicity cell. The two other panels in Figure~\ref{fig:abc-steady} show that the vortex rotates around its axis while the inner layers move at slightly faster speeds than its outer layers.

\begin{figure}[!ht]
  \centering
  \begin{subfigure}[t]{0.35\linewidth}\centering
    \includegraphics[height=40mm]{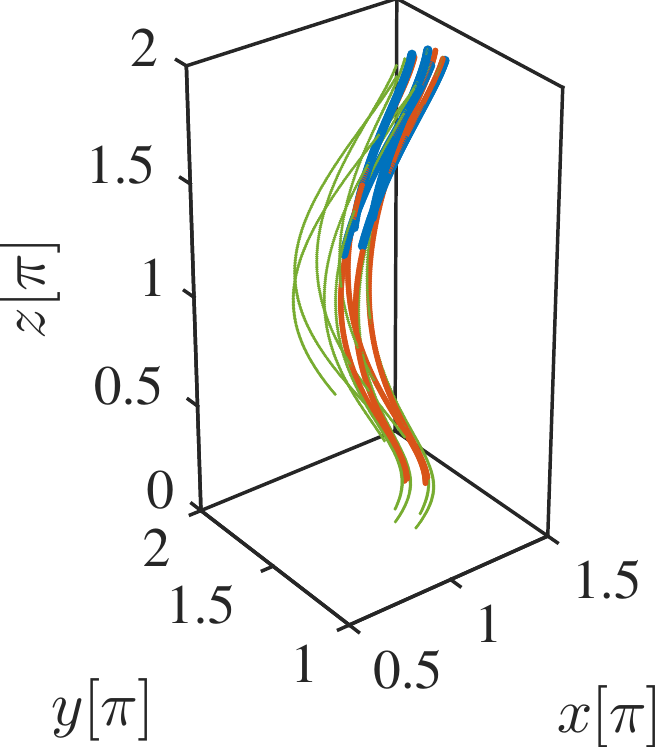}
    \caption{Pathlines for \(T = 1,\,2.5,\,5\) in a periodized unit cell}
  \end{subfigure}\quad
  \begin{subfigure}[t]{0.35\linewidth}\centering
    \includegraphics[height=40mm]{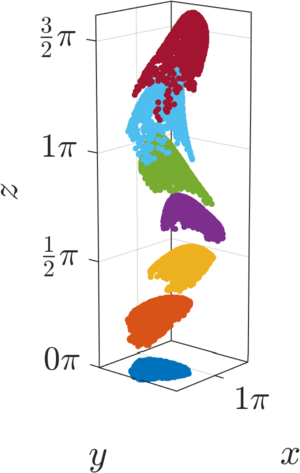}
    \caption{Advection of points for \(T = 0, 0.75, 1.50,\dots, 4.50\) (color indicates time).}\label{fig:abc-steady-cross}
  \end{subfigure}\\
  \begin{subfigure}[t]{0.7\linewidth}\centering
    \includegraphics[height=40mm]{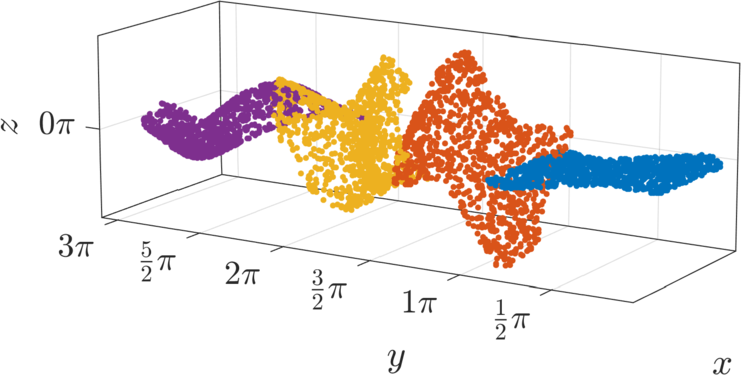}
    \caption{Advection of points for \(T = 0,  1.5,3.0, 4.5\) (color indicates time).}\label{fig:abc-steady-axial}
  \end{subfigure}\caption{\label{fig:abc-steady}Pathlines and advection patterns for the steady ABC
    flow~\protect\eqref{eq:abcmodel} with
    \(A=\sqrt{3},\,B=\sqrt{2},\,C=1,\,D=0\). Initial conditions clouds were
    sampled at \(t=0\) uniformly from the \(z=0\) section of the two vortices
    in Fig.~\protect\ref{fig:vortices-slice}; panel \protect\subref{fig:abc-steady-cross} samples
    the vortex centered at \(x = 7\pi/8,\,y=\pi/2\); panel
    \protect\subref{fig:abc-steady-axial} samples the vortex with the central axis
    along the line \(x = \pi/2\).}
\end{figure}

To detect non-mesohyperbolic behavior we need to numerically evaluate
conditions~\eqref{eq:mesohyperbolicity-timeT}
\begin{equation}\label{eq:hyp-criteria}
  t_{\psi} - t_{\psi^{-1}}  = 0, \quad\text{or}\quad t_{\psi} + t_{\psi^{-1}}  + 2 = 0,
\end{equation}
written here using the traces of the flow map and the relation \(m_{\psi} =
t_{\psi^{-1}}\) to estimate their growth more clearly. These conditions are
difficult to reliably compute in the face of numerical errors that will
almost-certainly result in non-zero quantities.

Numerically, we determine when the criteria are satisfied using a numerical
tolerance determined by estimating the growth rate of \(t_{\psi}\) and
\(t_{\psi^{-1}}\) with increasing length \(T\) of time
intervals. We infer the behavior of these quantities in linear
  time-invariant systems in order to establish the correct baseline --- the
  criterion has to reflect correctly our knowledge of linear systems if we are
  to even consider using it for classification of nonlinear systems. In what
  follows, we derive an empirical estimate of mesohyperbolicity, using a
  quantity termed \emph{numerical hyperbolicity}. Our considerations will
  further rely on \(T \to \infty\) arguments as our intent is to estimate the
  behavior of the flow beyond the time interval in which we are sampling it.

For a linear, time-invariant system, eigenvalues of the time-\(T\) map are
given either by \(\pm e^{\lambda_{1}T}\), \(\pm e^{\lambda_{2}T}\), \(\pm
e^{-(\lambda_{1}+\lambda_{2})T}\), or by \(\pm e^{\lambda T \pm i\omega T}\),
\(\pm e^{-\lambda T}\).

In both cases, as \(T \to \infty\),
\begin{equation}
    \abs{ t_{\psi} - t_{\psi^{-1}} } \sim e^{ \max_{i} \lambda_{i} T },\qquad
    \abs{ t_{\psi} + t_{\psi^{-1}} + 2} \sim e^{ \max_{i} \lambda_{i} T}
\end{equation}
To account for exponential growth, we set the numerical tolerance of
mesohyperbolicity based on logarithms of expressions~\eqref{eq:hyp-criteria}
\begin{equation}\label{eq:num-hyp}
    h_{1} := \frac{1}{T}\log \abs{ t_{\psi} - t_{\psi^{-1}} }, \qquad
    h_{2} := \frac{1}{T}\log \abs{ t_{\psi} + t_{\psi^{-1}} + 2}.
\end{equation}
(In all expressions we omit dependence on state variables and time interval
for shortness).

Non-zero values of either \(h_{1}\) or \(h_{2}\) are signs of
  mesohyperbolic behavior; conversely, we need either one of them to be small
  to declare non-mesohyperbolicity. In nonlinear flows, we do not expect that
  \(h_{1,2}\) will be entirely independent of the value of
  \(T\). Nevertheless, in ergodic regions,~\cite{Young1998} we expect
  convergence in mean as \(T\to\infty\).

  Even the rate of convergence to the mean is not uniform: in regular ergodic
  regions, e.g., vortices of the ABC flow, the expected decay is
  \(\mathcal{O}(T^{-1})\); in strongly mixing regions, conjectured to be
  embedded within the chaotic region, the expected decay is
  \(\mathcal{O}(T^{-1/2})\), i.e., similar to the Central Limit Theorem for
  i.i.d.~random variables. Initial conditions that are neither regular nor
  strongly mixing may potentially have an even slower decay of
  variance~\cite{Petersen1989}, \(T^{-\alpha}\) for any \(0 < \alpha <
  1/2\). Values of \(h\) at those points would then still grow as \(h \sim
  T^{1/2 - \alpha}\). The volume of weakly mixing zones is small in systems
  containing \kam{}-type dynamics,~\cite{Budisic2012b,Perry1994,Treschev1998}
  and therefore we do not expect those values to occur as major features in
  the histogram of \(h\).

A good quantitative criterion for deciding whether a point is
mesohyperbolic should estimate whether the smaller of the two \(h_{1,2}\)
\begin{equation}
\max_{i}\abs{\lambda_{i}} \sim \min\{h_{1},h_{2}\} \to 0\label{eq:decay-of-hyperbolicity}
\end{equation}
is ``sufficiently'' close to zero. Under a conjecture that some variant of the Central
Limit Theorem (CLT) holds for the estimate of \(\max_{i}\abs{\lambda_{i}}\) we
can test whether the deviation of our estimator \(\min\{h_{1},h_{2}\}\) from
the hypothesis of non-mesohyperbolicity \(\max_{i}\abs{\lambda_{i}} = 0\) is
normally distributed, i.e., whether \emph{numerical mesohyperbolicity} \(h\), defined by
\begin{equation}
\begin{aligned}
  h &= \abs{\min\{h_{1},h_{2}\}} \sqrt{T},\ \text{where, as before}\\
h_{1} &:= \frac{1}{T} \log\abs{d_{\tilde{f}}T^{3}}, \qquad
h_{2} := \frac{1}{T} \log\abs{8-2m_{\tilde{f}}T^{2}-3 d_{\tilde{f}}T^{3}}
\end{aligned}\label{eq:min-num-hyp}
\end{equation}
is small, \(h < \varepsilon\), for some (small) constant \(\varepsilon\). When this is indeed the case, then we are reasonably confident that with longer \(T\) the estimated eigenvalues of the flow map would indeed have a unit modulus and we empirically declare that the point is not mesohyperbolic, as classified by  Theorem~\ref{thm:mesochronic-classification}.

Is there a fixed value of \(\varepsilon\) that could help us decide whether the hypothesis of
mesohyperbolicity for a point holds? To answer this question standard CLT results from probability theory, e.g.,
Lindeberg--Feller~\cite[Thm.~3.4.5]{Durrett2010}, would employ higher order
moments of the distribution of samples of the random processes
\(h_{1,2}\). Such a constant \(\varepsilon\) would then turn our empirical
criterion into a statistical test, where \(h\) would be the z-score for
testing the hypothesis of whether a point under consideration is mesohyperbolic or
not. Unfortunately, we cannot assume to know how the higher moments behave,
despite existing work on CLTs in the context of dynamical
systems~\cite{Gouezel2014, Gouezel2004, Rey-Bellet2008, Young1998,
  Young1990}.

\begin{figure}[!ht]
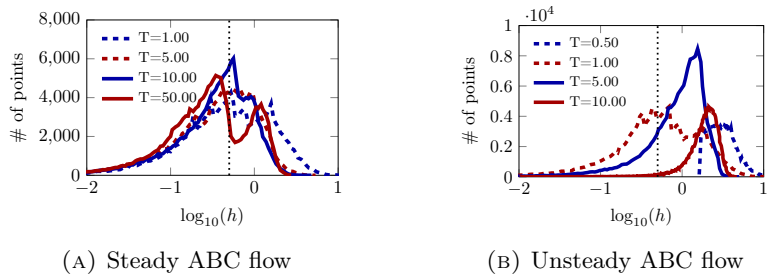

  \centering
  \begin{subfigure}[t]{0.45\linewidth}\centering
    \includegraphics[height=30mm]{img/tikz/ABC-steady-T-50_0-hyperdist.tikz}
    \caption{Steady ABC flow}\label{fig:dist-num-hyp-steady}
  \end{subfigure}
  \begin{subfigure}[t]{0.45\linewidth}\centering
    \includegraphics[height=30mm]{img/tikz/ABC-unsteady-T-10_0-hyperdist.tikz}
    \caption{Unsteady ABC flow}\label{fig:dist-num-hyp-unsteady}
  \end{subfigure}
  \caption{Distribution of numerical mesohyperbolicity (\protect\ref{eq:min-num-hyp})
    in steady and unsteady \abc{} flows for different lengths \(T\) of integration
    intervals, computed on a uniform grid of \(400\times 400\)
    points. Mesohyperbolicity was declared for \(h > 10^{-0.3} = 0.5\) (dotted vertical line). }\label{fig:dist-num-hyp}
\end{figure}

In lieu of rigorous results, we choose to proceed empirically and set the
cutoff value \(\varepsilon\) based on distributions of the numerical
mesohyperbolicity \(h\).  Figure~\ref{fig:dist-num-hyp-steady} shows histograms of
numerical mesohyperbolicity \(h\) for a range of values of \(T\), conforming
well to expectations.  As \(T\) increases, the distribution of \(h\) changes
from a fairly flat distribution (\(T=1\)) to a bimodal distribution with
well-separated peaks. Figure~\ref{fig:num-hyp} shows that each mode of
distribution of \(h\) corresponds, respectively, to vortices and to the large
chaotic region between them as \(T\to\infty\). Based on these results, we
declare numerical non-mesohyperbolicity using the cutoff parameter value
\begin{equation}
  \label{eq:num-cutoff}
  h < \varepsilon,\ \text{with}\ \varepsilon = 10^{-0.3}.
\end{equation}

\begin{figure}[!ht]
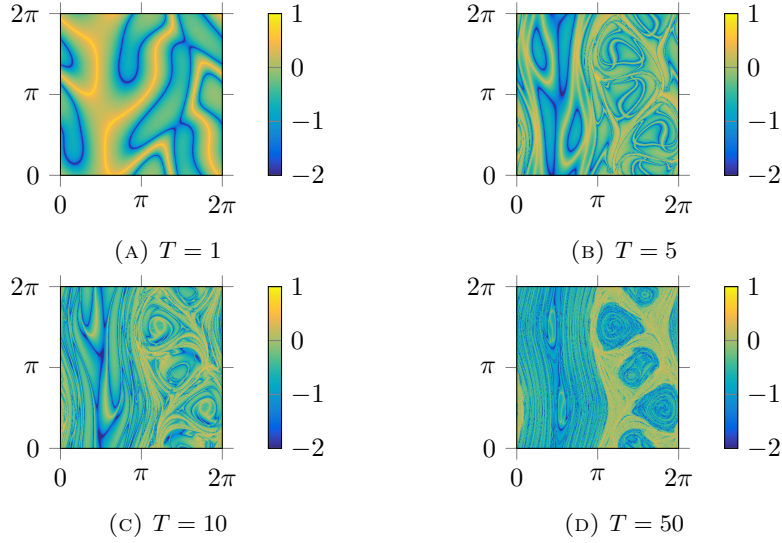

  \centering
  \begin{subfigure}[t]{0.45\linewidth}\centering
    \includegraphics[height=30mm]{img/tikz/ABC-steady-T-01_0-hyper.tikz}
    \caption{\(T=1\)}
  \end{subfigure}\quad
  \begin{subfigure}[t]{0.45\linewidth}\centering
    \includegraphics[height=30mm]{img/tikz/ABC-steady-T-05_0-hyper.tikz}
    \caption{\(T=5\)}
  \end{subfigure}\\
  \begin{subfigure}[t]{0.45\linewidth}\centering
    \includegraphics[height=30mm]{img/tikz/ABC-steady-T-10_0-hyper.tikz}
    \caption{\(T=10\)}
  \end{subfigure}\quad
  \begin{subfigure}[t]{0.45\linewidth}\centering
    \includegraphics[height=30mm]{img/tikz/ABC-steady-T-50_0-hyper.tikz}
    \caption{\(T=50\)}
  \end{subfigure}
  \caption{Spatial distribution of numerical
    mesohyperbolicity  on a plane in the state
    space of the steady \abc{} flow. Color is \(\log_{10} h\) with \(h\) defined as in~\protect\eqref{eq:min-num-hyp}.}\label{fig:num-hyp}
\end{figure}

The \owc{} criterion requires \(d_{f} \not =0 \) for non-hyperbolic sets. For
the given parameters at \(z = 0\),
\begin{align*}
  d_{f}(x,y,z) &= ABC( \cos x \cos y \cos z - \sin x \sin y \sin z)
                                                \\&= \sqrt{6}\cos x \cos y.
\end{align*}
Therefore, a solution \(\phi(\cdot,0,p)\) with \(p = (x,y,0)\) is
instantaneously hyperbolic everywhere except along the lines
\begin{equation}
  x = \frac{\pi}{2},\,
  x = \frac{3\pi}{2},\,
  y = \frac{\pi}{2},\,
  y = \frac{3\pi}{2}.\label{eq:abc-ow-nonhyp}
\end{equation}

\begin{figure}[!ht]
  \centering
    \includegraphics[height=40mm]{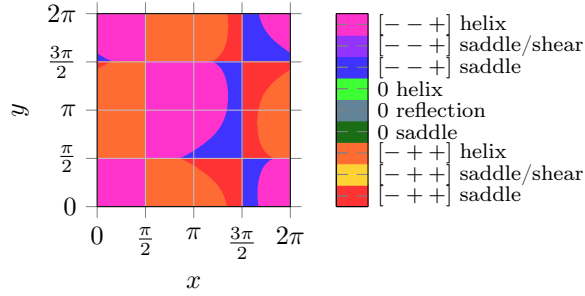}\\
  \caption{Mesochronic classes in the \(x-y\) plane of \(z=0\) slice of the \abc{} flow
    state space for \(T = 10^{-2}\). As \(T \approx 0\), the mesochronic partition is virtually identical to the \owc{} partition.}\label{fig:mesohyp-small-T}
\end{figure}
The mesochronic partition for \(T \approx 0\) is illustrated in
Figure~\ref{fig:mesohyp-small-T}, and due to short integration period \(T\),
matches exactly the \owc{} partition. Notice that the conventional intuition
about vortices being ``elliptic'' structures cannot be inferred from short
integration times, as for \(T \approx 0\) almost the entire space is
mesohyperbolic (non-mesohyperbolic lines~\eqref{eq:abc-ow-nonhyp} are
difficult to sample numerically). Nevertheless, neutral mesohelical regions roughly
coincide with locations of vortices, while the chaotic region between vortices
contains a mixture of all four classes of mesohyperbolicity. Comparing with
Figure~\ref{fig:abcflow-statespace}, we see that the boundaries of mesochronic
classes do not align with boundaries of invariant sets.

\begin{figure}[!ht]
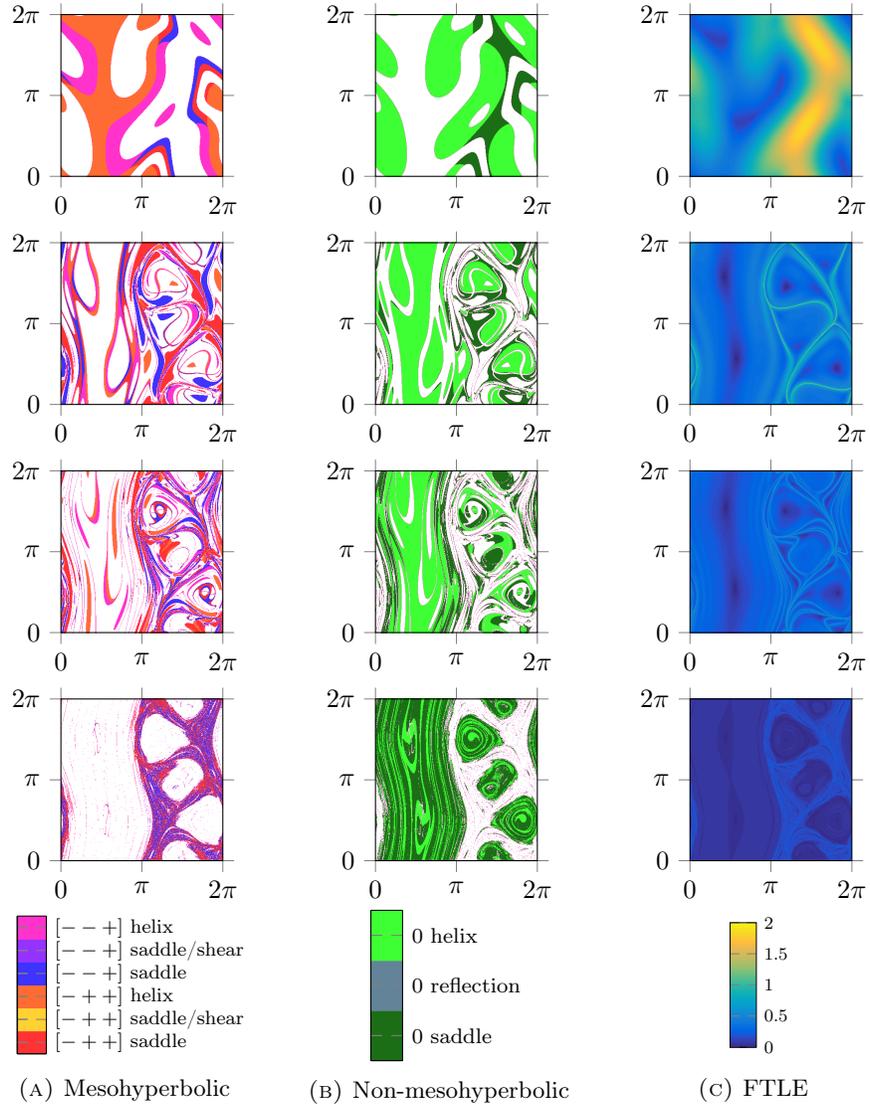

  \centering
  \begin{subfigure}[t]{0.32\linewidth}\centering
    \includegraphics[height=30mm]{img/tikz/ABC-steady-T-01_0-hyclass.tikz}\\
    \includegraphics[height=30mm]{img/tikz/ABC-steady-T-05_0-hyclass.tikz}\\
    \includegraphics[height=30mm]{img/tikz/ABC-steady-T-10_0-hyclass.tikz}\\
    \includegraphics[height=30mm]{img/tikz/ABC-steady-T-50_0-hyclass.tikz}\\
    \includegraphics[height=20mm]{img/tikz/colorbar-hy.tikz}
\caption{Mesohyperbolic}
  \end{subfigure}
  \begin{subfigure}[t]{0.32\linewidth}\centering
    \includegraphics[height=30mm]{img/tikz/ABC-steady-T-01_0-heclass.tikz}\\
    \includegraphics[height=30mm]{img/tikz/ABC-steady-T-05_0-heclass.tikz}\\
    \includegraphics[height=30mm]{img/tikz/ABC-steady-T-10_0-heclass.tikz}\\
    \includegraphics[height=30mm]{img/tikz/ABC-steady-T-50_0-heclass.tikz}\\
    \includegraphics[height=20mm]{img/tikz/colorbar-he.tikz}
\caption{Non-mesohyperbolic}
  \end{subfigure}
  \begin{subfigure}[t]{0.32\linewidth}\centering   \includegraphics[height=30mm]{img/tikz/ABC-steady-T-01_0-ftle.tikz}\\
    \includegraphics[height=30mm]{img/tikz/ABC-steady-T-05_0-ftle.tikz}\\
    \includegraphics[height=30mm]{img/tikz/ABC-steady-T-10_0-ftle.tikz}\\
    \includegraphics[height=30mm]{img/tikz/ABC-steady-T-50_0-ftle.tikz}\\
    \includegraphics[height=20mm]{img/tikz/colorbar-ftle.tikz}
\caption{FTLE}
  \end{subfigure}
  \caption{Distribution of mesochronic classes and the \ftle{} (FTLE) field on the \(x-y\) plane of \(z=0\) slice of the \emph{steady} ABC flow for times (in rows) \(T = 1,\,5,\,10,\,50\). }\label{fig:mesohyp-steady-T}
\end{figure}

Increasing \(T\) results in the sequence of images shown in
Figure~\ref{fig:mesohyp-steady-T}, where the mesohyperbolic regions are shown in the left
column, and the non-mesohyperbolic regions in the right column. As the averaging
period is increased to \(T = 1\), partitions deform, but remain largely
uncorrelated with invariant features. As we increase \(T\) beyond \(1\),
non-hyperbolic behavior significantly re-appears along the interface between
different mesohyperbolic classes. Parts of boundaries of mesohyperbolic zones
start to align with invariant vortices. Since level sets of any function
averaged for a sufficiently long time will partition the state space into
invariant sets~\cite{Mezic1999,Budisic2012b}, this is expected. Notice that the
non-mesohyperbolic regions appear almost exclusively inside invariant
vortices. As \(T\) is increased even further, the non-mesohyperbolic zones
grow inside the invariant vortices and eventually completely occupy them. In
the chaotic zone, we see disappearance of mesohelical dynamics, with only
saddle mesohyperbolicity remaining, which matches the asymptotic analysis in Section~\ref{sec:limit-to-infty}.
\clearpage
\subsubsection{Unsteady case}
\label{sec:unsteady-ABC}

We now keep the parameters \(A=\sqrt{3},\,B=\sqrt{2},\,C=1\) as before, but
set \(D = 1\), which results in the unsteady variation in the coefficient
\(A(t) = \sqrt{3} + t \sin t\). The unsteady ABC flow has not received as much
analytic attention as the steady case; however, it was used as a demonstration
of numerical techniques for computation of the flow map Jacobian \(\jac
\psi_{T}\) in~\cite{Brunton2010}. Figure~\ref{fig:abc-unsteady} shows the same sets of initial conditions used in demonstrating the steady flow (Figure~\ref{fig:abc-steady}), but now advected by the unsteady flow for comparison. Notice that the initial advection patterns are similar, until \(A(t)\) starts significantly deviating from its constant term \(A\).

\begin{figure}[!ht]
  \centering
  \begin{subfigure}[t]{0.35\linewidth}\centering
    \includegraphics[height=40mm]{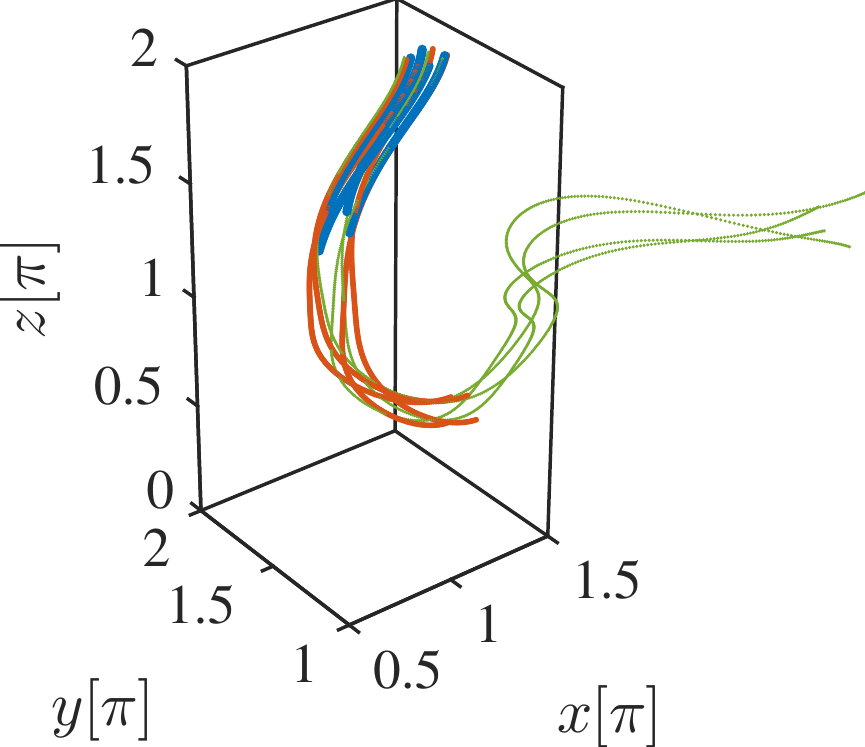}
    \caption{Pathlines for \(T = 1,\,2.5,\,5\) in a periodized unit cell}
  \end{subfigure}\quad
  \begin{subfigure}[t]{0.35\linewidth}\centering
    \includegraphics[height=40mm]{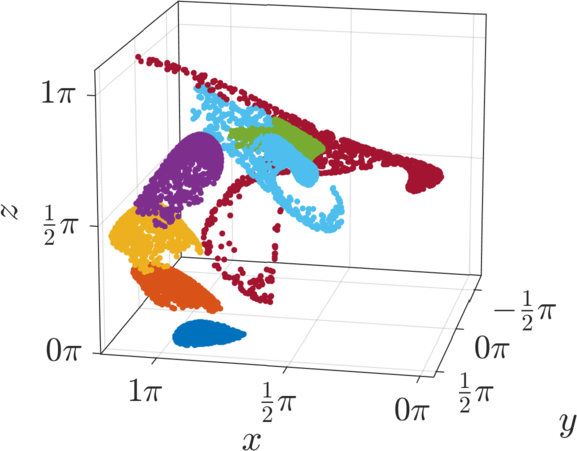}
    \caption{Advection of points for \(T = 0,\, 0.75,\, 1.50,\dots,\, 4.50\) (color indicates time).}\label{fig:abc-unsteady-cross}
  \end{subfigure}\\
  \begin{subfigure}[t]{0.7\linewidth}\centering
    \includegraphics[height=40mm]{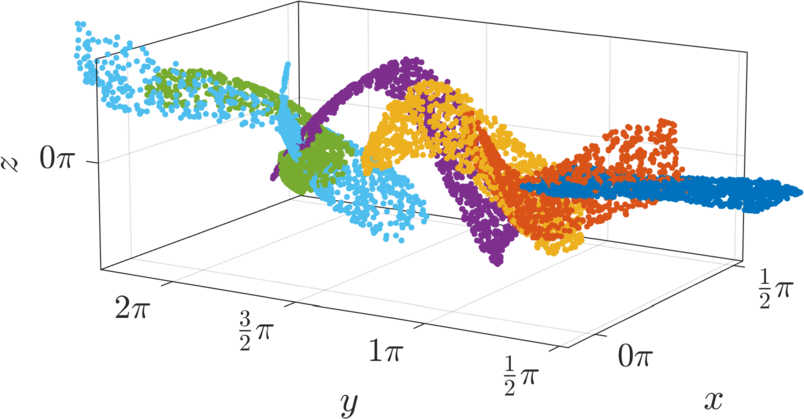}
    \caption{Advection of points for \(T = 0,\,1.5,\,3.0,\, 4.5\) (color indicates time).}\label{fig:abc-unsteady-axial}
  \end{subfigure}
  \caption{\label{fig:abc-unsteady}Pathlines and advection patterns for the unsteady ABC flow~\protect\eqref{eq:abcmodel} with \(A=\sqrt{3},\,B=\sqrt{2},\,C=1,\,D=1\). Initial conditions clouds were sampled at \(t=0\) uniformly from the \(z=0\) section of the two vortices in Fig.~\protect\ref{fig:vortices-slice}; panel \protect\subref{fig:abc-unsteady-cross} samples the vortex centered at \(x = 7\pi/8,\,y=\pi/2\); panel \protect\subref{fig:abc-unsteady-axial} samples the vortex with the central axis along the line \(x = \pi/2\).}
\end{figure}

Figure~\ref{fig:dist-num-hyp-unsteady} shows the histogram of numerical mesohyperbolicity while  Figure~\ref{fig:mesohyp-unsteady-T} shows the spatial distributions of
(non-)mesohyperbolicity classes, determined using the same numerical criterion as in
the steady case~\eqref{eq:min-num-hyp}. For short intervals \(T = 0.5,\,
1.0\), mesochronic classification of the flow is similar to the steady
case. This is expected as the magnitude of \(A(t)\) is dominated by the steady
component. As time evolves, non-mesohyperbolic regions in the flow are
destroyed, and the obvious split between two behaviors that was observed in
the steady flow (Figure~\ref{fig:mesohyp-steady-T}) is not present
here. Remnants of the axial vortex in the left and two ``eyes'' of vortices in
the right sides of images are visible both in mesochronic and FTLE partitions.
\begin{figure}[!ht]
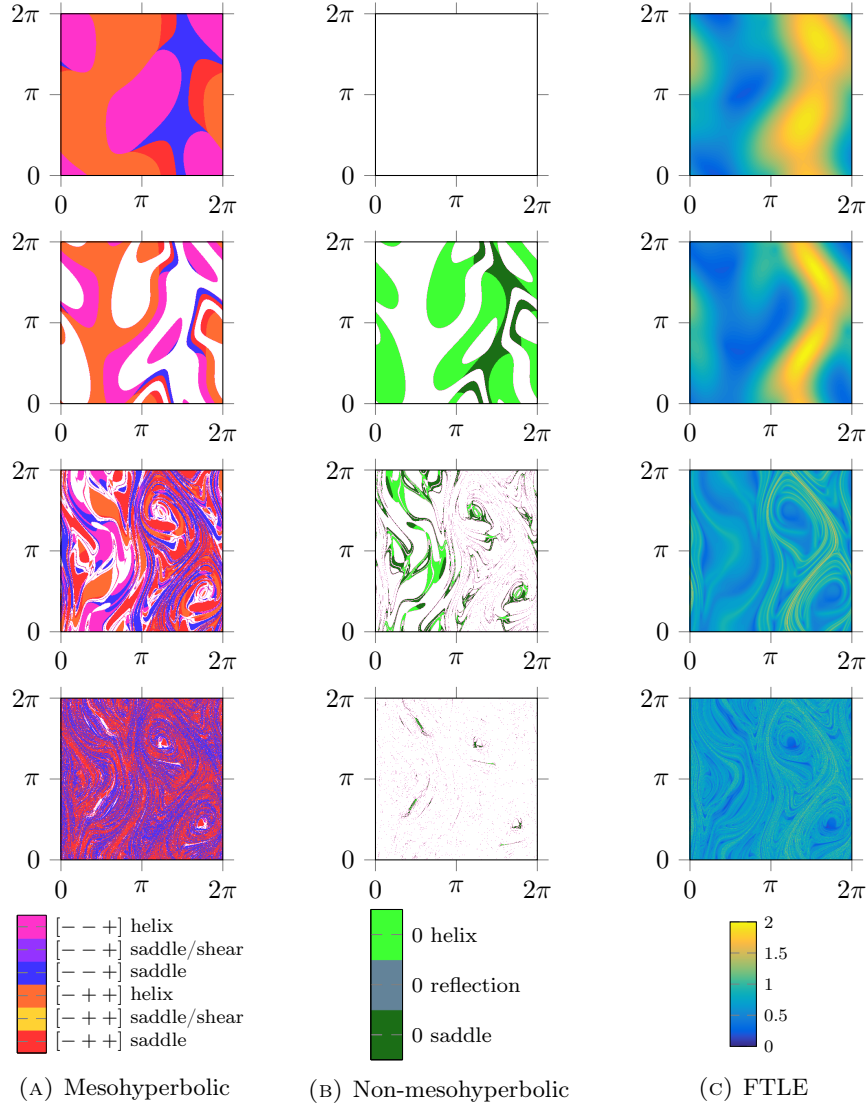

  \centering
  \begin{subfigure}[t]{0.32\linewidth}\centering
    \includegraphics[height=30mm]{img/tikz/ABC-unsteady-T-00_5-hyclass.tikz}\\
    \includegraphics[height=30mm]{img/tikz/ABC-unsteady-T-01_0-hyclass.tikz}\\
    \includegraphics[height=30mm]{img/tikz/ABC-unsteady-T-05_0-hyclass.tikz}\\
    \includegraphics[height=30mm]{img/tikz/ABC-unsteady-T-10_0-hyclass.tikz}\\
    \includegraphics[height=20mm]{img/tikz/colorbar-hy.tikz}
\caption{Mesohyperbolic}\label{fig:unsteady-hy}
  \end{subfigure}
  \begin{subfigure}[t]{0.32\linewidth}\centering
    \includegraphics[height=30mm]{img/tikz/ABC-unsteady-T-00_5-heclass.tikz}\\
    \includegraphics[height=30mm]{img/tikz/ABC-unsteady-T-01_0-heclass.tikz}\\
    \includegraphics[height=30mm]{img/tikz/ABC-unsteady-T-05_0-heclass.tikz}\\
    \includegraphics[height=30mm]{img/tikz/ABC-unsteady-T-10_0-heclass.tikz}\\
    \includegraphics[height=20mm]{img/tikz/colorbar-he.tikz}
\caption{Non-mesohyperbolic}
  \end{subfigure}
  \begin{subfigure}[t]{0.32\linewidth}\centering
    \includegraphics[height=30mm]{img/tikz/ABC-unsteady-T-00_5-ftle.tikz}\\
    \includegraphics[height=30mm]{img/tikz/ABC-unsteady-T-01_0-ftle.tikz}\\
    \includegraphics[height=30mm]{img/tikz/ABC-unsteady-T-05_0-ftle.tikz}\\
    \includegraphics[height=30mm]{img/tikz/ABC-unsteady-T-10_0-ftle.tikz}\\
    \includegraphics[height=20mm]{img/tikz/colorbar-ftle.tikz}
\caption{FTLE}
  \end{subfigure}
  \caption{Distribution of mesochronic classes and the \ftle{} (FTLE) field on the \(x-y\) plane of \(z=0\) slice of the \emph{unsteady} ABC flow for times (in rows) \(T = 1,\,5,\,10,\,50\). }\label{fig:mesohyp-unsteady-T}
\end{figure}

\begin{figure}[htb]
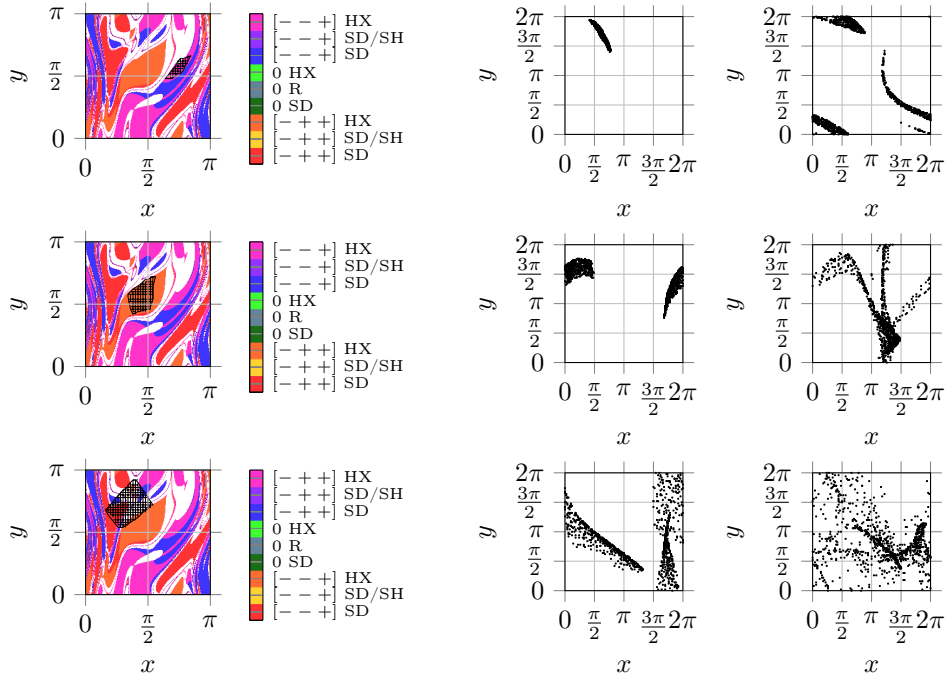

  \centering
  \begin{subfigure}[t]{0.45\linewidth}
    \includegraphics[height=30mm]{img/tikz/mmpHX_patch_t00.tikz}\\
    \includegraphics[height=30mm]{img/tikz/mppHX_patch_t00.tikz}\\
    \includegraphics[height=30mm]{img/tikz/straddle_patch_t00.tikz}
\caption{Initial patch (hatched) at \(t = t_{0} = 0\)}
  \end{subfigure}\hfill
  \begin{subfigure}[t]{0.25\linewidth}
    \includegraphics[height=30mm]{img/tikz/mmpHX_patch_t05.tikz}\\
    \includegraphics[height=30mm]{img/tikz/mppHX_patch_t05.tikz}\\
    \includegraphics[height=30mm]{img/tikz/straddle_patch_t05.tikz}
\caption{Tracer at \(t=5\)}
  \end{subfigure}~
  \begin{subfigure}[t]{0.25\linewidth}
    \includegraphics[height=30mm]{img/tikz/mmpHX_patch_t07.tikz}\\
    \includegraphics[height=30mm]{img/tikz/mppHX_patch_t07.tikz}\\
    \includegraphics[height=30mm]{img/tikz/straddle_patch_t07.tikz}
\caption{Tracer at \(t=7\)}
  \end{subfigure}
  \caption{Material advection in the unsteady ABC flow. Rows show clouds of \(10^{3}\) points from regions that are, respectively, mesohelical \([- - +]\), \([- - +]\) and mixed-mesochronic-class between \(t=0\) and \(t=5\). The first column shows the selection patch at time \(t=0\), overlaid from the lower-left square of the third (\(T=5\)) panel in  Fig.~\protect\ref{fig:unsteady-hy}, with two additional times shown in the adjoining columns. Point clouds are graphed as they project onto \(z=0\) plane.}
  \label{fig:unsteady-material-advection}
\end{figure}

Figure~\ref{fig:unsteady-material-advection} illustrates transport of initial conditions sampled from several regions in the state space of the unsteady ABC flow. The first two rows show advection from patches chosen as subsets of regions that are mesohyperbolic for integration times \(T = 0.5, 1, 5\). The last row shows advection from a patch that straddles several mesochronic regions at \(T = 5\).

Advection up to \(t = 5\) demonstrates that the initial conditions selected from single mesochronic regions (central column, top two frames) do not disperse much, i.e., stay coherent. Initial conditions from a mixture of regions show considerably more dispersion.

Advection up to \(t = 7\), shown in the third column, demonstrates how the patches of initial conditions evolve \emph{beyond} the interval \(T = 5\) that was used to generate mesochronic classes. All material patches at this point show considerable growth; arguably, the patch in the last row again shows the largest dispersion.

\begin{figure}[htb]
  \centering
  \begin{subfigure}[t]{0.45\linewidth}
    \includegraphics[height=40mm]{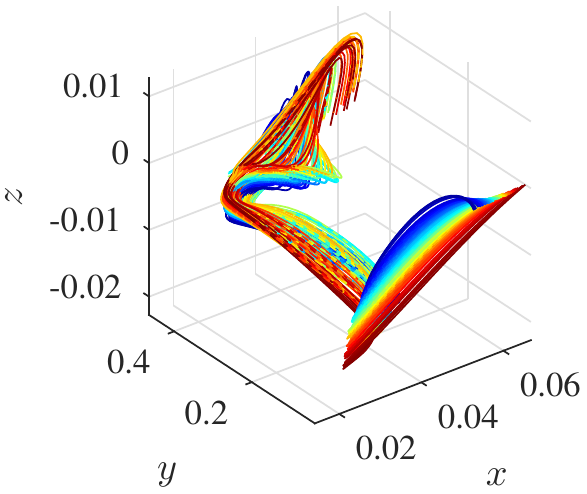}
    \caption{\(\bm{[- - +]}\) mesohelical patch}
  \end{subfigure}\quad
  \begin{subfigure}[t]{0.45\linewidth}
    \includegraphics[height=40mm]{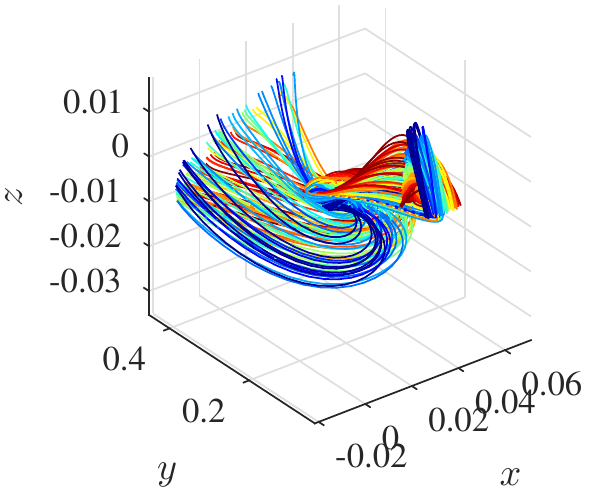}
    \caption{\(\bm{[- + +]}\) mesohelical patch}
  \end{subfigure}\\
  \begin{subfigure}[t]{0.45\linewidth}
    \includegraphics[height=40mm]{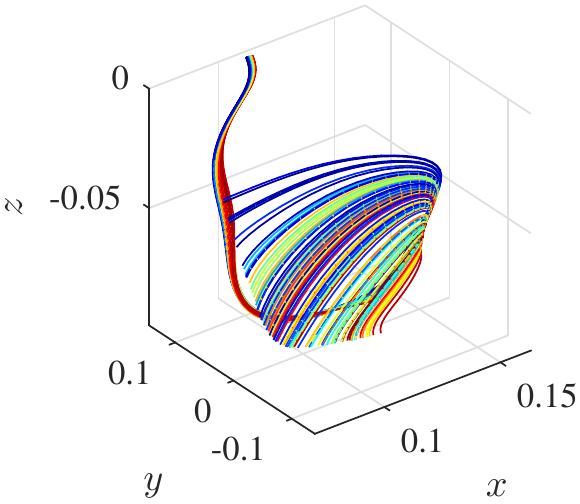}
    \caption{\(\bm{[- - +]}\) mesosellar patch}
  \end{subfigure}\quad
  \begin{subfigure}[t]{0.45\linewidth}
    \includegraphics[height=40mm]{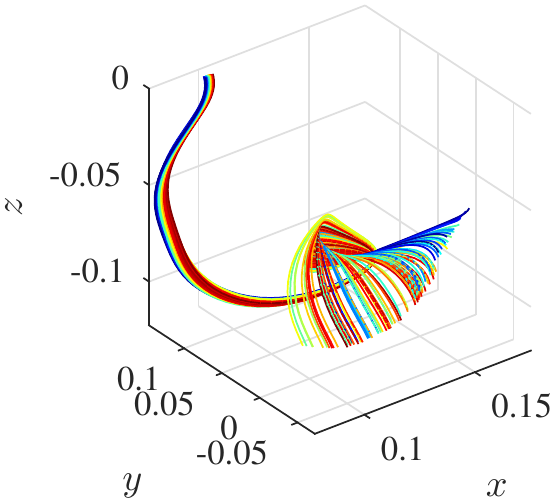}
    \caption{\(\bm{[- + +]}\) mesosellar patch}
  \end{subfigure}\caption{\label{fig:unsteady-pathlines}Orbits (pathlines) of \(100\) points initialized from two mesohelical patches shown in Figure~\protect\ref{fig:unsteady-material-advection} and two additional mesosellar patches, advected for time \(t \in [0,5]\). Color is added to illustrate internal ordering of the material.}
\end{figure}

We conclude this section with Figure~\ref{fig:unsteady-pathlines} showing
images of orbits, i.e., pathlines, of material points advected by the unsteady
ABC flow through space over 5 time units. Two mesohelical patches used to
initialize points are the same as in
Figure~\ref{fig:unsteady-material-advection}; the two additional mesosellar
patches were initialized similarly. The two mesosellar sets are sampled from a
narrower region, so it is not surprising that they remain more tightly packed
than mesohelical sets. We believe that the major distinction between the top
and the bottom row is in the internal order of orbits within each
bundle. Mesohelical sets seem to preserve the internal order, i.e., despite
the torsion of the bulk, one can still observe an ordered rainbow of colors
toward the end of the orbits. On the other hand, both mesosellar sets show
more internal mixing of colors. While this is potentially a circumstantial
observation here, it may be an interesting point to explore in the future, due
to the association of hyperbolic saddles with mixing and bulk rotation with
lack of mixing.

\section{Discussion}
\label{sec:discussion}

Mesochronic analysis is a local analysis in that it classifies a single
point based on the properties of the local deformation gradients
along trajectory emanating from it. Nevertheless, the hope is that sets of initial conditions selected
using the local mesochronic criteria will collectively stay coherent on a
macroscopic level and, potentially, deform as a bulk in the way suggested by
the local analysis. The contributions of this paper are in extension of
two-dimensional theory of~\cite{Mezic2010} to three-dimensional flows.

We have extended the concepts of mesohyperbolicity (local strain over finite
times) and mesoellipticity (local rotation over finite times) to three
dimensions, which allows for co-existence of the rotation and the strain. In
turn, we differentiate between the mesosellar behavior, involving three
distinct directions of strain, and the mesohelical behavior, involving a plane
in which material simultaneously rotates and, possibly, uniformly
strains. Quantities \(\Sigma\) and \(\Delta\) that we defined simplify
classification of domain points, side-stepping explicit calculation of
eigenvalues of the flow. The eigenvectors of $\jac \psi_T$ also have a role in
the foregoing analysis. A complex conjugate pair determines a plane whose
normal is the vector of finite-time rotation. A real eigenvector indicates the
direction of finite-time stretching under the map.

Rotation of the material has a much richer presence in 3D than in 2D
classification. In two dimensions, any non-uniformity in strain, accompanied
by rotation or not, manifests itself as a hyperbolic deformation. In other
words, only rigid-body rotations are highlighted as non-hyperbolic in two
dimensions, in addition to shear. In three dimensions, rotation in a plane can
be accompanied by strain in the normal direction: the type of that strain
distinguishes between three mesohelical classes: \([- + +],\ [- - +]\) and
neutral.

Consequently, the invariant vortices in the steady ABC flow initially show a
significant amount of hyperbolicity in them, indicating that layers within
them move at different speeds in the axial direction
(Figure~\ref{fig:abc-steady}); however, over the longer timescales they appear
neutrally mesohelical, which corresponds to interpretation of them as rigid
rotating structures. This suggests that the separation between different
layers of vortices is asymptotically sub-exponential.

In the unsteady ABC flow discussed in Section~\ref{sec:unsteady-ABC}, material
is significantly less coherent as the unsteadyness destroys long-term
invariant structures. Nevertheless, for interval lengths \(T = 1,\ 5\), while
the unsteady term \(A(t) = A + Dt\sin(t)\) is of the same order of magnitude
as the steady \(A\), the distributions of mesohyperbolic areas show loose
correspondence between steady and unsteady ABC flows. However, as \(T \to
10\), the regions that turned neutrally mesohelical in the steady case are
instead replaced by a mixture of hyperbolic mesosellar regions, indicative of
chaotic mixing.

We have demonstrated that visualization of mesochronic classes corresponds
with well-known behavior of the fluid-like ABC flow. In particular, it is
interesting to see how well the mesohelical regions correspond with the vortex
zones in which \kam{}-type structures are known to exist. At this point, the
theory is immediately applicable to kinematic analysis of the geometric
structure in chaotic advection in fluids; however, numerical algorithms used
in this paper serve the proof-of-concept purpose, and more accurate and
efficient methods for computation of the flow Jacobian~\cite{Brunton2010} can
and should be readily used, if applicable.

Note that the mesochronic classification is invariant to Galilean
transformations but not to rotating frame transformations~\cite{Haller2015},
it therefore discovers transport properties for three-dimensional
incompressible flows which are observed in the given frame of reference the
system is in. It is interesting to consider the relationship between the
notion of exponential dichotomy~\cite{Coppel1978} and mesohyperbolicity
discussed here. Mesohyperbolicity (or -ellipticity, -helicity) are notions
that are valid on finite time intervals.  They are best used as tools to study
the bifurcations of local dynamics of a trajectory in time, when changing time
intervals are selected, i.e. mesohyperbolicity is a function of two
parameters, the beginning and final time. It appears plausible that, adapting
the techniques used in~\cite{Palmer2011}, it can be shown that if a trajectory
is not mesohyperbolic for any interval $[t_1,t_2]$ inside the interval of
interest $[T_1,T_2],$ then there is no (appropriately
defined~\cite{Palmer2011}) finite-time exponential dichotomy on that
interval. Thus, the concepts defined here have the potential to parametrize
finite-time stability properties.

The true test of any method for detection of geometric structures is its
usefulness to the applied communities, e.g., physical oceanography, and flow
engineers. We therefore plan to apply the technique to the unsteady testbed
flows in~\cite{Balasuriya2012}, to transitory systems~\cite{Mosovsky2011}
and to more physically-relevant models in order to further verify practical
use of the mesochronic classification, beyond confirmations obtained for the
\twod{} case. Furthermore, it remains to be understood if the highlighted
quantities, e.g., the trace, the cofactor trace, and determinant of the
mesochronic Jacobian, are also useful in dynamic analysis of turbulent fluid
flows.

\bibliography{references}

\providecommand{\href}[2]{#2}
\providecommand{\arxiv}[1]{\href{http://arxiv.org/abs/#1}{arXiv:#1}}
\providecommand{\url}[1]{\texttt{#1}}
\providecommand{\urlprefix}{URL }
\begin{thebibliography}{10}

\bibitem{Adrianova1995}
\newblock L.~Y. Adrianova,
\newblock \emph{Introduction to linear systems of differential equations}, vol.
  146 of Translations of Mathematical Monographs,
\newblock {American Mathematical Society, Providence, RI}, 1995.

\bibitem{Allshouse2012}
\newblock M.~Allshouse and J.-L. Thiffeault,
\newblock Detecting coherent structures using braids,
\newblock \emph{Physica D. Nonlinear Phenomena}, 95--105.

\bibitem{Aref1984a}
\newblock H.~Aref and E.~P. Flinchem,
\newblock Dynamics of a {{Vortex Filament}} in a {{Shear-Flow}},
\newblock \emph{Journal of Fluid Mechanics}, \textbf{148} (1984), 477--497.

\bibitem{Balasuriya2012}
\newblock S.~Balasuriya,
\newblock Explicit invariant manifolds and specialised trajectories in a class
  of unsteady flows,
\newblock \emph{Physics of Fluids (1994-present)}, \textbf{24} (2012), 127101.

\bibitem{Blazevski2014}
\newblock D.~Blazevski and G.~Haller,
\newblock Hyperbolic and elliptic transport barriers in three-dimensional
  unsteady flows,
\newblock \emph{Physica D: Nonlinear Phenomena}, \textbf{273{\textendash}274}
  (2014), 46--62.

\bibitem{Boyland2000}
\newblock P.~L. Boyland, H.~Aref and M.~A. Stremler,
\newblock Topological fluid mechanics of stirring,
\newblock \emph{Journal of Fluid Mechanics}, \textbf{403} (2000), 277--304.

\bibitem{Brunton2010}
\newblock S.~L. Brunton and C.~W. Rowley,
\newblock Fast computation of finite-time {{Lyapunov}} exponent fields for
  unsteady flows,
\newblock \emph{Chaos: An Interdisciplinary Journal of Nonlinear Science},
  \textbf{20} (2010), 017503.

\bibitem{Budisic2012b}
\newblock M.~Budi{\v s}i{\'c} and I.~Mezi{\'c},
\newblock Geometry of the ergodic quotient reveals coherent structures in
  flows,
\newblock \emph{Physica D. Nonlinear Phenomena}, \textbf{241} (2012),
  1255--1269.

\bibitem{Chong1990}
\newblock M.~S. Chong, A.~E. Perry and B.~J. Cantwell,
\newblock A general classification of three-dimensional flow fields,
\newblock \emph{Physics of Fluids A: Fluid Dynamics (1989-1993)}, \textbf{2}
  (1990), 765--777.

\bibitem{Coppel1978}
\newblock W.~A. Coppel,
\newblock \emph{Dichotomies in stability theory},
\newblock Lecture Notes in Mathematics, Vol. 629, {Springer-Verlag, Berlin-New
  York}, 1978.

\bibitem{Dellnitz1999}
\newblock M.~Dellnitz and O.~Junge,
\newblock On the approximation of complicated dynamical behavior,
\newblock \emph{SIAM Journal on Numerical Analysis}, \textbf{36} (1999),
  491--515.

\bibitem{Dellnitz2002}
\newblock M.~Dellnitz and O.~Junge,
\newblock Set oriented numerical methods for dynamical systems,
\newblock in \emph{Handbook of dynamical systems, {{Vol}}. 2},
\newblock {North-Holland}, Amsterdam, 2002,
\newblock 221--264.

\bibitem{Dombre1986}
\newblock T.~Dombre, U.~Frisch, J.~M. Greene, M.~H{\'e}non, A.~Mehr and A.~M.
  Soward,
\newblock Chaotic streamlines in the {{ABC}} flows,
\newblock \emph{Journal of Fluid Mechanics}, \textbf{167} (1986), 353--391.

\bibitem{Durrett2010}
\newblock R.~Durrett,
\newblock \emph{Probability: theory and examples},
\newblock 4th edition,
\newblock Cambridge Series in Statistical and Probabilistic Mathematics,
  {Cambridge University Press, Cambridge}, 2010.

\bibitem{Farazmand2015a}
\newblock M.~Farazmand,
\newblock Hyperbolic {{Lagrangian}} coherent structures align with contours of
  path-averaged scalars,
\newblock \emph{arXiv:1501.05036 {[}nlin, physics:physics]}.

\bibitem{Farazmand2016}
\newblock M.~Farazmand and G.~Haller,
\newblock Polar rotation angle identifies elliptic islands in unsteady
  dynamical systems,
\newblock \emph{Physica D: Nonlinear Phenomena}, \textbf{315} (2016), 1--12.

\bibitem{Fox2013}
\newblock A.~M. Fox and J.~D. Meiss,
\newblock Greene's residue criterion for the breakup of invariant tori of
  volume-preserving maps,
\newblock \emph{Physica D: Nonlinear Phenomena}, \textbf{243} (2013), 45--63.

\bibitem{Froyland2003}
\newblock G.~Froyland and M.~Dellnitz,
\newblock Detecting and locating near-optimal almost-invariant sets and cycles,
\newblock \emph{SIAM Journal on Scientific Computing}, \textbf{24} (2003),
  1839--1863 (electronic).

\bibitem{Froyland2010a}
\newblock G.~Froyland, S.~Lloyd and N.~Santitissadeekorn,
\newblock Coherent sets for nonautonomous dynamical systems,
\newblock \emph{Physica D: Nonlinear Phenomena}, \textbf{239} (2010),
  1527--1541.

\bibitem{Froyland2009}
\newblock G.~Froyland and K.~Padberg,
\newblock Almost-invariant sets and invariant manifolds{\textemdash}connecting
  probabilistic and geometric descriptions of coherent structures in flows,
\newblock \emph{Physica D. Nonlinear Phenomena}, \textbf{238} (2009),
  1507--1523.

\bibitem{Froyland2010}
\newblock G.~Froyland, N.~Santitissadeekorn and A.~Monahan,
\newblock Transport in time-dependent dynamical systems: {{Finite}}-time
  coherent sets,
\newblock \emph{Chaos: An Interdisciplinary Journal of Nonlinear Science},
  \textbf{20} (2010), 043116.

\bibitem{Gelfand2008}
\newblock I.~M. Gelfand, M.~M. Kapranov and A.~V. Zelevinsky,
\newblock \emph{Discriminants, resultants and multidimensional determinants},
\newblock Modern Birkh{\"a}user Classics, {Birkh{\"a}user Boston Inc.}, Boston,
  MA, 2008.

\bibitem{Goldhirsch1987}
\newblock I.~Goldhirsch, P.-L. Sulem and S.~A. Orszag,
\newblock Stability and {{Lyapunov}} stability of dynamical systems: {{A}}
  differential approach and a numerical method,
\newblock \emph{Physica D: Nonlinear Phenomena}, \textbf{27} (1987), 311--337.

\bibitem{Gouezel2004}
\newblock S.~Gou{\"e}zel,
\newblock Central limit theorem and stable laws for intermittent maps,
\newblock \emph{Probability Theory and Related Fields}, \textbf{128} (2004),
  82--122.

\bibitem{Gouezel2014}
\newblock S.~Gou{\"e}zel and I.~Melbourne,
\newblock Moment bounds and concentration inequalities for slowly mixing
  dynamical systems,
\newblock \emph{Electronic Journal of Probability}, \textbf{19} (2014), no. 93,
  30.

\bibitem{Greene1968}
\newblock J.~M. Greene,
\newblock Two-dimensional measure-preserving mappings,
\newblock \emph{Journal of Mathematical Physics}, \textbf{9} (1968), 760--768.

\bibitem{Greene1979}
\newblock J.~M. Greene,
\newblock Method for {{Determining}} a {{Stochastic Transition}},
\newblock \emph{Journal of Mathematical Physics}, \textbf{20} (1979),
  1183--1201.

\bibitem{Haller2001}
\newblock G.~Haller,
\newblock Lagrangian structures and the rate of strain in a partition of
  two-dimensional turbulence,
\newblock \emph{Physics of Fluids}, \textbf{13} (2001), 3365--3385.

\bibitem{Haller2011a}
\newblock G.~Haller,
\newblock A variational theory of hyperbolic {{Lagrangian Coherent
  Structures}},
\newblock \emph{Physica D. Nonlinear Phenomena}, \textbf{240} (2011), 574--598.

\bibitem{Haller2015}
\newblock G.~Haller,
\newblock Lagrangian {{Coherent Structures}},
\newblock \emph{Annual Review of Fluid Mechanics}, \textbf{47} (2015),
  137--162.

\bibitem{Haller2012}
\newblock G.~Haller and F.~J. Beron-Vera,
\newblock Geodesic theory of transport barriers in two-dimensional flows,
\newblock \emph{Physica D: Nonlinear Phenomena}, \textbf{241} (2012),
  1680--1702.

\bibitem{Haller1998}
\newblock G.~Haller and A.~C. Poje,
\newblock Finite time transport in aperiodic flows,
\newblock \emph{Physica D. Nonlinear Phenomena}, \textbf{119} (1998), 352--380.

\bibitem{Haller2000}
\newblock G.~Haller and G.~Yuan,
\newblock Lagrangian coherent structures and mixing in two-dimensional
  turbulence,
\newblock \emph{Physica D. Nonlinear Phenomena}, \textbf{147} (2000), 352--370.

\bibitem{Irving2004}
\newblock R.~S. Irving,
\newblock \emph{Integers, polynomials, and rings},
\newblock Undergraduate Texts in Mathematics, {Springer-Verlag}, New York,
  2004.

\bibitem{Koopman1931}
\newblock B.~O. Koopman,
\newblock Hamiltonian {{Systems}} and {{Transformations}} in {{Hilbert Space}},
\newblock \emph{Proceedings of National Academy of Sciences}, \textbf{17}
  (1931), 315--318.

\bibitem{Levnajic2010}
\newblock Z.~Levnaji{\'c} and I.~Mezi{\'c},
\newblock Ergodic theory and visualization. {{I}}. {{Mesochronic}} plots for
  visualization of ergodic partition and invariant sets,
\newblock \emph{Chaos: An Interdisciplinary Journal of Nonlinear Science},
  \textbf{20} (2010), --.

\bibitem{Ma2014}
\newblock T.~Ma and E.~M. Bollt,
\newblock Differential {{Geometry Perspective}} of {{Shape Coherence}} and
  {{Curvature Evolution}} by {{Finite-Time Nonhyperbolic Splitting}},
\newblock \emph{SIAM Journal on Applied Dynamical Systems}, \textbf{13} (2014),
  1106--1136.

\bibitem{Ma2015}
\newblock T.~Ma and E.~M. Bollt,
\newblock Shape {{Coherence}} and {{Finite-Time Curvature Evolution}},
\newblock \emph{International Journal of Bifurcation and Chaos}, \textbf{25}
  (2015), 1550076.

\bibitem{Ma2016}
\newblock T.~Ma, N.~T. Ouellette and E.~M. Bollt,
\newblock Stretching and folding in finite time,
\newblock \emph{Chaos: An Interdisciplinary Journal of Nonlinear Science},
  \textbf{26} (2016), 023112.

\bibitem{Madrid2009}
\newblock J.~A.~J. Madrid and A.~M. Mancho,
\newblock Distinguished trajectories in time dependent vector fields,
\newblock \emph{Chaos: An Interdisciplinary Journal of Nonlinear Science},
  \textbf{19} (2009), 013111.

\bibitem{Malhotra1998}
\newblock N.~Malhotra, I.~Mezi{\'c} and S.~Wiggins,
\newblock Patchiness: {{A New Diagnostic}} for {{Lagrangian Trajectory
  Analysis}} in {{Time-Dependent Fluid Flows}},
\newblock \emph{International Journal of Bifurcation and Chaos}, \textbf{08}
  (1998), 1053--1093.

\bibitem{Mancho2013}
\newblock A.~M. Mancho, S.~Wiggins, J.~Curbelo and C.~Mendoza,
\newblock Lagrangian descriptors: {{A}} method for revealing phase space
  structures of general time dependent dynamical systems,
\newblock \emph{Communications in Nonlinear Science and Numerical Simulation},
  \textbf{18} (2013), 3530--3557.

\bibitem{Mezic1994}
\newblock I.~Mezi{\'c},
\newblock \emph{On the geometrical and statistical properties of dynamical
  systems: theory and applications},
\newblock Phd thesis, California Institute of Technology, 1994.

\bibitem{Mezic2005}
\newblock I.~Mezi{\'c},
\newblock Spectral properties of dynamical systems, model reduction and
  decompositions,
\newblock \emph{Nonlinear Dynamics. An International Journal of Nonlinear
  Dynamics and Chaos in Engineering Systems}, \textbf{41} (2005), 309--325.

\bibitem{Mezic2004}
\newblock I.~Mezi{\'c} and A.~Banaszuk,
\newblock Comparison of systems with complex behavior,
\newblock \emph{Physica D. Nonlinear Phenomena}, \textbf{197} (2004), 101--133.

\bibitem{Mezic2010}
\newblock I.~Mezi{\'c}, S.~Loire, V.~A. Fonoberov and P.~J. Hogan,
\newblock A {{New Mixing Diagnostic}} and {{Gulf Oil Spill Movement}},
\newblock \emph{Science Magazine}, \textbf{330} (2010), 486--489.

\bibitem{Mezic2002}
\newblock I.~Mezi{\'c} and F.~Sotiropoulos,
\newblock Ergodic theory and experimental visualization of invariant sets in
  chaotically advected flows,
\newblock \emph{Physics of Fluids}, \textbf{14} (2002), 2235.

\bibitem{Mezic1999}
\newblock I.~Mezi{\'c} and S.~Wiggins,
\newblock A method for visualization of invariant sets of dynamical systems
  based on the ergodic partition,
\newblock \emph{Chaos: An Interdisciplinary Journal of Nonlinear Science},
  \textbf{9} (1999), 213--218.

\bibitem{Mosovsky2011}
\newblock B.~A. Mosovsky and J.~D. Meiss,
\newblock Transport in transitory dynamical systems,
\newblock \emph{SIAM Journal on Applied Dynamical Systems}, \textbf{10} (2011),
  35--65.

\bibitem{Nocedal2006}
\newblock J.~Nocedal and S.~J. Wright,
\newblock \emph{Numerical optimization},
\newblock 2nd edition,
\newblock Springer Series in Operations Research and Financial Engineering,
  {Springer}, New York, 2006.

\bibitem{Okubo1970}
\newblock A.~Okubo,
\newblock Horizontal {{Dispersion}} of {{Floatable Particles}} in {{Vicinity}}
  of {{Velocity Singularities Such}} as {{Convergences}},
\newblock \emph{Deep-Sea Research}, \textbf{17} (1970), 445--454.

\bibitem{Ottino1989}
\newblock J.~M. Ottino,
\newblock \emph{The kinematics of mixing: stretching, chaos, and transport},
\newblock Cambridge Texts in Applied Mathematics, {Cambridge University Press},
  Cambridge, 1989.

\bibitem{Palmer2011}
\newblock K.~J. Palmer,
\newblock A finite-time condition for exponential dichotomy,
\newblock \emph{Journal of Difference Equations and Applications}, \textbf{17}
  (2011), 221--234.

\bibitem{Perry1994}
\newblock A.~D. Perry and S.~Wiggins,
\newblock {{KAM}} tori are very sticky: rigorous lower bounds on the time to
  move away from an invariant {{Lagrangian}} torus with linear flow,
\newblock \emph{Physica D. Nonlinear Phenomena}, \textbf{71} (1994), 102--121.

\bibitem{Petersen1989}
\newblock K.~Petersen,
\newblock \emph{Ergodic theory}, vol.~2 of Cambridge Studies in Advanced
  Mathematics,
\newblock {Cambridge University Press}, Cambridge, UK, 1989.

\bibitem{Poje1999}
\newblock A.~Poje, G.~Haller and I.~Mezi{\'c},
\newblock The geometry and statistics of mixing in aperiodic flows,
\newblock \emph{Physics of Fluids}, \textbf{11} (1999), 2963--2968.

\bibitem{Rey-Bellet2008}
\newblock L.~Rey-Bellet and L.-S. Young,
\newblock Large deviations in non-uniformly hyperbolic dynamical systems,
\newblock \emph{Ergodic Theory and Dynamical Systems}, \textbf{28} (2008),
  587--612.

\bibitem{Ruelle1985}
\newblock D.~P. Ruelle,
\newblock Rotation numbers for diffeomorphisms and flows,
\newblock \emph{Annales de l'Institut Henri Poincar{\'e}. Physique
  Th{\'e}orique}, \textbf{42} (1985), 109--115.

\bibitem{Samelson2013}
\newblock R.~M. Samelson,
\newblock Lagrangian {{Motion}}, {{Coherent Structures}}, and {{Lines}} of
  {{Persistent Material Strain}},
\newblock \emph{Annual Review of Marine Science}, \textbf{5} (2013), 137--163.

\bibitem{Shadden2005}
\newblock S.~C. Shadden, F.~Lekien and J.~E. Marsden,
\newblock Definition and properties of {{Lagrangian}} coherent structures from
  finite-time {{Lyapunov}} exponents in two-dimensional aperiodic flows,
\newblock \emph{Physica D. Nonlinear Phenomena}, \textbf{212} (2005), 271--304.

\bibitem{SzezechJr2013}
\newblock J.~D. {Szezech Jr.}, A.~B. Schelin, I.~L. Caldas, S.~R. Lopes, P.~J.
  Morrison and R.~L. Viana,
\newblock Finite-time rotation number: {{A}} fast indicator for chaotic
  dynamical structures,
\newblock \emph{Physics Letters A}, \textbf{377} (2013), 452--456.

\bibitem{Thiffeault2010}
\newblock J.-L. Thiffeault,
\newblock Braids of entangled particle trajectories,
\newblock \emph{Chaos: An Interdisciplinary Journal of Nonlinear Science},
  \textbf{20} (2010), 017516--017514.

\bibitem{Treschev1998}
\newblock D.~Treschev,
\newblock Width of stochastic layers in near-integrable two-dimensional
  symplectic maps,
\newblock \emph{Physica D. Nonlinear Phenomena}, \textbf{116} (1998), 21--43.

\bibitem{Weiss1991}
\newblock J.~Weiss,
\newblock The dynamics of enstrophy transfer in two-dimensional hydrodynamics,
\newblock \emph{Physica D. Nonlinear Phenomena}, \textbf{48} (1991), 273--294.

\bibitem{Wiggins1992}
\newblock S.~Wiggins,
\newblock \emph{Chaotic transport in dynamical systems}, vol.~2 of
  Interdisciplinary Applied Mathematics,
\newblock {Springer-Verlag}, New York, 1992.

\bibitem{Young1990}
\newblock L.-S. Young,
\newblock Some {{Large Deviation Results}} for {{Dynamical Systems}},
\newblock \emph{Transactions of the American Mathematical Society},
  \textbf{318} (1990), 525--543.

\bibitem{Young1998}
\newblock L.-S. Young,
\newblock Statistical {{Properties}} of {{Dynamical Systems}} with {{Some
  Hyperbolicity}},
\newblock \emph{Annals of Mathematics}, \textbf{147} (1998), 585--650.

\end{thebibliography}

\appendix

\section{Cubic polynomials and \texorpdfstring{\(3 \times 3\)}{
    3 x 3
  } matrices}
\label{sec:cubic-poly-mat}

In this section we review formulas and notation for \(3 \times 3\) matrices.
Let \(A\) denote a \(3 \times 3\) matrix over the field \(\R\), \(
A =
\left[\begin{smallmatrix}
    a_{11} & a_{12} & a_{13} \\
    a_{21} & a_{22} & a_{23} \\
    a_{31} & a_{32} & a_{33}
  \end{smallmatrix}\right]
\).
The characteristic polynomial of \(A\) is
\begin{equation}
  \label{eq:char-poly}
  p(\lambda) := \det( \lambda \cdot \id - A),
\end{equation}
where \(\id\) is the unit matrix. The zeros of \(p(\lambda)\) are the eigenvalues of \(A\). Define the following quantities:
\begin{equation}
  \begin{aligned}
    t_{A} :=& \tr A, \quad d_{A} := \det A,\\
    m_{A} :=& \tr \cof A =
    \det \left[\begin{smallmatrix}
        a_{11}& a_{12} \\
        a_{21}& a_{22}
      \end{smallmatrix}\right] \\
    &+
    \det \left[\begin{smallmatrix}
        a_{11}& a_{13} \\
        a_{31}& a_{33}
      \end{smallmatrix}\right]
    +
    \det \left[\begin{smallmatrix}
        a_{22}& a_{23} \\
        a_{32}& a_{33}
      \end{smallmatrix}\right],\\
  \end{aligned}\label{eq:poly-coeff}
\end{equation}
where \(\cof A\) is the cofactor matrix of \(A\), containing signed minors.\footnote{The transpose of \(\cof A\) is the adjugate matrix of \(A\).}
By expanding equation~\eqref{eq:char-poly} we can verify that
\begin{equation}\label{eq:char-poly-matrix}
  p(\lambda)=\lambda^3-t_{A}\lambda^2+m_{A}\lambda-d_{A}.
\end{equation}
For a polynomial with real coefficients, the zeros \(\lambda_{1}, \lambda_{2}, \lambda_{3}\) are either all real, or two of them appear as a complex conjugate pair, in which case we denote them by \(\lambda_{r}, \lambda_{c}, \bar \lambda_{c}\). By convention, we will assume that the imaginary part \(\Im \lambda_{c} > 0\).

The \emph{discriminant} of a cubic polynomial
\(f(\lambda) = \lambda^{3} + f_{2} \lambda^{2} + f_{1}\lambda + f_{0}\)
is given by
\begin{equation*}
  \begin{aligned}
    \mathcal{D}(f) :=&\ 18f_{2}f_{1}f_{0} - 4f_{2}^{3}f_{0} + f_{2}^{2}f_{1}^{2} \\
    &- 4 f_{1}^{3} - 27f_{0}^{2},\label{eq:discriminant-cubic}
  \end{aligned}
\end{equation*}
its sign determines whether the zeros of \(f\) lie on the real axis or not (see Ref.~\cite{Gelfand2008}, \S 12.1.B).

Using the notation from~\eqref{eq:char-poly-matrix}, we see that \(f_{2} = -t_{A}\), \(f_{1} = m_{A}\), \(f_{0} = d_{A}\), therefore
\begin{equation*}
  \label{eq:discriminant-cubic-matrix}
  \begin{aligned}
    \mathcal{D}(p) =& 18 t_{A} m_{A} d_{A} - 4 t_{A}^{3} d_{A} \\
    &+ t_{A}^{2}m_{A}^{2} -4 m_{A}^{3} -27 d_{A}^{2}.
  \end{aligned}
\end{equation*}
It then holds (see Ref.~\cite{Irving2004}, \S 10.3, Ex.\ 10.14) that:
\begin{compactitem}
\item \(\mathcal{D}(p) > 0\), the zeros are \(\lambda_{1},\lambda_{2},\lambda_{3} \in \R\),
\item \(\mathcal{D}(p) = 0\), the zeros are \(\lambda_{1}, \lambda_{2} = \lambda_{3} \in \R\),
\item \(\mathcal{D}(p) < 0\), the zeros are \(\lambda_{r} \in \R\), \(\lambda_{c}, \bar \lambda_{c} \in \C\).
\end{compactitem}

The following also holds in general for eigenvalues \(\lambda_{1,2,3} \in \C\), which can be seen by expanding \((\lambda-\lambda_{1})(\lambda-\lambda_{2})(\lambda-\lambda_{3})\) and comparing to~\eqref{eq:char-poly-matrix}:
\begin{align*}
  d_{A} =&\  \det A = \lambda_{1} \lambda_{2} \lambda_{3}, \\
  m_{A} =&\  \tr \cof A = \lambda_{1} \lambda_{2} + \lambda_{1}\lambda_{3} + \lambda_{2}\lambda_{3}, \\
  t_{A} =&\  \tr A = \lambda_{1} + \lambda_{2} + \lambda_{3}.
\end{align*}

\section{Differential equation for the mesochronic Jacobian}
\label{sec:ode-mh-jacobian}

If an initial condition \(x \in \R\) and initial time \(t_{0}\) are fixed, the mesochronic Jacobian matrix \(\jac \tilde{f}_{\tau}(x)\) satisfies a matrix-valued ODE in the variable \(\tau\), the length of the averaging interval. To evaluate the mesochronic Jacobian for the purposes of classifying \((t_{0},x)\) into one of the classes described in Section~\ref{sec:mesochronic}, we numerically solve a particular matrix-valued initial value problem and evaluate its solution at \(\tau = T\).

To derive the ODE for the mesochronic Jacobian, start from the integral expression for the time-\(\tau\) map and compute its Jacobian:
\begin{align*}
  &\psi_{\tau}(x) = x + \int_{t_{0}}^{t_{0} + \tau} f(t_{0}+t, \psi_{t}(x))dt \\
  &\jac \psi_{\tau}(x) = \id + \int_{t_{0}}^{t_{0} + \tau} \jac f(t_{0}+t, \psi_{t}(x)) \cdot \jac \psi_{t}(x) dt \\
  &\jac \psi_{\tau}'(x) = \jac f(t_{0}+\tau, \psi_{\tau}(x)) \cdot \jac \psi_{\tau}(x),
\end{align*}
where \('\) denotes \(d/d\tau\).

We now substitute relation~\eqref{eq:timeT-mh-j-link}, \(\jac \psi_{\tau}(x) = \id + \tau \jac \tilde{f}_{\tau}(x)\), linking the time-\(\tau\) map \(\psi_{\tau}\) and the mesochronic Jacobian \(\jac \tilde{f}_{\tau}\). To simplify notation, we use \(M(\tau) := \jac \tilde{f}_{\tau}(x)\) and \(A(\tau) := \jac f(t_{0}+\tau, \psi_{\tau}(x))\) for the mesochronic and advected velocity field Jacobian, respectively.
\begin{equation}
  \begin{aligned}
    &\left( \id + \tau M(\tau)\right)' = A(\tau) \cdot \left( \id + \tau M(\tau)\right) \\
    &M(\tau) + \tau \dot{M}(\tau) = A(\tau) + \tau A(\tau) \cdot  M(\tau) \\
    &\dot{M}(\tau) = \big(A(\tau) - M(\tau)\big)/\tau + A(\tau) \cdot  M(\tau).
  \end{aligned}\label{eq:mhj-ode}
\end{equation}
The initial condition for the matrix ODE is set at \(\tau = 0\) when the mesochronic Jacobian \(\jac \tilde{f}_{\tau}(x)\) is identical to the velocity field Jacobian
\(
\jac \tilde{f}_{0}(x) = \jac f(t_{0}, x),
\)
therefore
\begin{equation}
  M(0) = A(0)  \quad\text{and}\quad \dot{M}(0) = {A(0)}^{2},\label{eq:mhj-ode-ic}
\end{equation}
where the last calculation is obtained from~\eqref{eq:mhj-ode} by evaluating the first-order expansion of \(A(\tau) - M(\tau)\) at \(\tau = 0\).

\section{Implementation of the mesochronic classification}
\label{sec:implementation}

In what follows we provide the basic algorithm which we used to produce a
numerical implementation of mesochronic classification, used to generate
images analogous to Figures~\ref{fig:mesohyp-steady-T}, and~\ref{fig:mesohyp-unsteady-T}. The core task is to assign a
mesochronic class (Definition~\ref{def:spectral-classes}), corresponding to
the time interval \([t_{0},t_{0}+T]\), to a fixed initial condition \(x \in
\R^{3}\).

This task can be split into the following sequence of stages.\footnote{Dependence on values of \(x,\,t_{0},\,t_{0}+T\) is omitted, for shortness.}
\begin{compactenum}
\item\label{step:traj} Compute the trajectory segment \(x(t)\) for \(t \in [t_{0},t_{0}+T]\) where \(x(t_{0}) = x\).
\item\label{step:jac}  Evaluate the Jacobian matrix \(A(\tau) := \jac f(t_{0}+\tau, x(t_{0}+\tau))\) of the velocity field along the trajectory \(x(t)\).
\item\label{step:mhj} Evaluate the mesochronic Jacobian matrix \(\jac \tilde{f}_{\tau}\) at the endpoint \(\tau = T\) by integrating the initial value problem given by~\eqref{eq:mhj-ode} and~\eqref{eq:mhj-ode-ic}:
  \begin{align*}
    \dot{M} (\tau) &= \big(A(\tau) - M(\tau)\big)/\tau + A(\tau) \cdot  M(\tau),\\
    M(0) &= A(0),\
     \dot{M}(0) = A{(0)}^{2}.
  \end{align*}
\item\label{step:dettrace} Compute the determinant \(d_{\tilde{f}} = \det M(T)\), cofactor trace \(m_{\tilde{f}} = \tr \cof M(T)\), evaluate \(\Delta\) and \(\Sigma\),~\eqref{eq:sigmadelta}.
\item\label{step:class} Based on the signs of \(\Delta\) and \(\Sigma\) assign
  the mesochronic class to the pair of initial condition and time interval
  \((p,[t_{0},t_{0}+T])\) using Theorem~\ref{thm:mesochronic-classification}.
\end{compactenum}

\begin{remarkstep*}[\ref{step:traj}] We assume that the velocity field
\(f(t,x)\) can be evaluated on the entire domain of interest and throughout
the entire compact time interval \(I \in \R\). If the velocity field is known
only on a grid of points \((t,x)\), then interpolation of the velocity field is likely needed.
\end{remarkstep*}
\begin{remarkstep*}[\ref{step:jac}] We assume the knowledge of the
Jacobian matrix of the vector field evaluated along the trajectory segment
\([t_{0},t_{0}+\tau]\), termed \(A(\tau) := \jac f(t_{0}+\tau,
x(t_{0}+\tau))\). If the velocity field was known analytically, it might be
possible to express \(\jac f\) analytically as well and evaluate it along
points computed in Step~\ref{step:traj}. If not, \(\jac f\) can be numerically
approximated using a central spatial-difference scheme with a spatial step
\(\delta\). On finite-precision computers, \(\delta\) cannot be
taken arbitrarily small, due to finite-precision effects. There will always
exist an optimal, non-zero \(\delta\), which depends on the machine
precision and magnitude of higher derivatives of \(f\) (see Ref.~\cite{Nocedal2006}, \S
8).
\end{remarkstep*}
\begin{remarkstep*}[\ref{step:mhj}] Equation~\eqref{eq:mhj-ode} is a
linear, matrix-valued ODE where the vector field Jacobian \(A(t)\), computed
in Step~\ref{step:jac}, comes in as both inhomogeneity and as the parametric
term. This ODE can be discretized using one of the standard time-stepping
schemes, on the same time points used for discretization of \(x(t)\) in Step
(\ref{step:traj}). The examples in this paper were computed using an explicit
Adams--Bashforth stepping scheme with a fixed time step.
\end{remarkstep*}
Since each of the presented steps involves some degree of numerical
approximation, the mesochronic Jacobian \(M(T)\) will contain numerical
noise. As the derivation of the mesochronic classification criteria hinges on
the assumption of incompressibility of the flow,
i.e.,~\eqref{eq:incompressibility-flow}, the necessary criterion for
accuracy of the numerical approximation is that the numerical compressibility
\(\tr M(\tau) + \tau \cdot \tr \cof M(\tau) +\tau^{2} \det M(\tau) \approx
0\).  If the numerical compressibility is significantly larger than \(0\) when
evaluated using numerical \(M(\tau)\), this can be taken as an indication of
significant numerical errors in the mesochronic Jacobian and, consequently, in
mesochronic classification.


\end{document}